\documentclass[runningheads]{llncs}
\usepackage[T1]{fontenc}

\usepackage{enumitem}
\usepackage{tikz}
\usepackage{todonotes}
\usetikzlibrary{shapes.misc, decorations.markings, calc}

\tikzset{cross/.style={cross out, draw=black, minimum size=2*(#1-\pgflinewidth), inner sep=0pt, outer sep=0pt},
cross/.default={2.7pt}}

\usepackage{amsmath}
\usepackage{amsfonts}
\usepackage{amssymb}
\usepackage{mathtools}

\spnewtheorem{protocol}{Protocol}{\bfseries}{\itshape}

\spnewtheorem*{mainthm}{Theorem~\ref{th:main}}{\normalfont\bfseries}{\itshape}
\spnewtheorem*{sndthm}{Theorem~\ref{th:snd}}{\normalfont\bfseries}{\itshape}

\usepackage[hypertexnames=false]{hyperref}
\usepackage{color}

\DeclareMathOperator{\subop}{sub}
\newcommand{\subc}{\subop}

\newif\ifextended
\extendedtrue
\begin{document}

%\title{The Solvability of Approximate Agreement on Simplicial Complexes}
%\title{Solving Approximate Agreement on CUB Spaces and Simplicial Complexes}
%\newcommand{\mytitle}{Solving Approximate Agreement on CUB Spaces and Simplicial Complexes}
\newcommand{\mytitle}{Solving Approximate Agreement on continuous and discrete spaces}
%\\ or\\ On which space can approximate agreement be solved?}

\title{\mytitle}

\author{Augustin Albert\inst{1} \and Sergio Rajsbaum\inst{2}}

\institute{LIX, CNRS, École polytechnique, Institut Polytechnique de Paris, Palaiseau, France
\email{augustin.albert@polytechnique.edu} \and
IRIF-CNRS-Universit\'e Paris Cit\'e on leave from Instituto de Matemáticas, UNAM,  Mexico,
\email{sergio.rajsbaum@gmail.com}\thanks{Support by UNAM-PAPIIT  37-IN109926} }

\authorrunning{Augustin Albert \and Sergio Rajsbaum}

\titlerunning{\mytitle}

\maketitle

\begin{abstract}

We consider $n$ asynchronous processes prone to crashes, communicating via shared read-write registers, and study the wait-free solvability of approximate agreement: given inputs, processes must output values that are close to each other, and satisfy a validity property (at the very least, if inputs are identical, all outputs must equal that input).

The problem has been studied for various input spaces: continuous, discrete, one-dimensional or multidimensional. For metric spaces, validity requires outputs to lie in the convex hull of the inputs. For graphs, and more generally simplicial complexes, several conditions exist.  We focus on simplex agreement and validity: outputs must span a simplex, and if inputs span a simplex $\sigma$ then outputs are in $\sigma$.

Solvability depends on the input space, validity condition, and number of processes. For example, the problem is solvable for all $n$ in the plane, but only for $n \leq 2$ when removing a point. For a graph, solvability for $n=2$ holds iff the graph is connected, but $n\geq 3$ requires acyclicity. 

In the discrete case, we prove that simplex agreement on a simplicial complex $\mathcal{C}$ is solvable for $n+1$ processes iff $\mathcal{C}$ is $(n-1)$-connected. We discuss several consequences, including a proof of a conjecture by Ledent.

In the continuous setting, we consider CUB spaces: a broad class of metric spaces admitting a unique convexity definition, subsuming classical $\epsilon$-agreement on $[0,1]$ and $m$-dimensional approximate agreement. Our results show that $\epsilon$-agreement is solvable in every CUB space.

\end{abstract}

\section{Introduction}

Approximate agreement has been thoroughly studied due to its theoretical and practical importance, in almost any model of distributed computing. It has recently attracted renewed interest from a wide range of application areas, including avionics~\cite{stolz2016byzantine}, robot coordination~\cite{herlihy2013distributed,alcantara2019topology}, %Potop2011
and distributed learning~\cite{su2016fault}.%,el2021collaborative}.

We consider the problem in a system of $n$ asynchronous processes, prone to crash failures, communicating through read-write shared registers. We concentrate on the \emph{wait-free} case (any number of processes may crash) due to its theoretical importance.

In the original problem, $\epsilon$-approximate agreement~\cite{dolev1986reaching}, processes are each given as input a point in $\mathbb R$, and after communicating a finite number of rounds, each process must output a point in the convex hull of the inputs, such that outputs are $\epsilon$-close to each other. Since then, the problem has been studied under a variety of input spaces: continuous and discrete, one-dimensional and multidimensional. For instance, $m$-dimensional approximate agreement~\cite{mendes2015multidimensional} considers input values in $\mathbb R^m$, and generalized approximate agreement~\cite{herlihy1993asynchronous} in a subset of $\mathbb R^m$. More recently, approximate agreements on a graph $G$~\cite{alcantara2019topology,castaneda2018convergence,NowakR19,Ledent21,Liu22,alistarh2023wait} consider as inputs the vertices of $G$. 
While $\mathbb R^m$, and more generally Euclidean spaces, are equipped with canonical notions of distance and convexity, this is not always the case on discrete structures like graphs. This has led to the definition of different \emph{validity} and \emph{agreement} conditions. For instance, edge~\cite{castaneda2018convergence} (resp. clique~\cite{alcantara2019topology}) agreement uses the following conditions:
\begin{itemize}[noitemsep]
    \item Validity: if the input vertices span an edge (resp. clique), then the output vertices are among the inputs.
    \item Agreement: output vertices span an edge (resp. clique).
\end{itemize}

We  focus on simplex agreement~\cite{Ledent21}, another discrete instance of approximate agreement, on simplicial complexes, the higher-dimensional counterpart to graphs. In a simplicial complex one may think of the vertices belonging to a simplex as being close to each other, in the same way as in a graph an edge represents nearness of its vertices.
\begin{definition}[Simplex agreement~\cite{Ledent21}\footnote{Note that this task is unrelated to the simplex agreement task considered in~\cite{HerlihyShavitJACM}.}]
Let $\mathcal C$ be a finite non-empty simplicial complex. Each process is given an input vertex $x_i \in \mathcal C$, and must output a vertex $y_i \in \mathcal C$, such that the following two conditions are satisfied:
\begin{itemize}[noitemsep]
\item Validity: if the $(x_i)_{i \in [n]}$ span a simplex $\sigma$ of $\mathcal C$, then all $y_i$ are in $\sigma$.
\item Agreement: the $(y_i)_{i \in [n]}$ span a simplex of $\mathcal C$.
\end{itemize}
\end{definition} 

Simplex agreement generalizes several other approximate agreement problems, most notably barycentric agreement~\cite{herlihy2013distributed}, edge and clique agreement. Indeed, any graph $G$ can be regarded as a $1$-dimensional complex, or as its clique complex $\operatorname{Cl(G)}$. Simplex agreement on $G$ and $\operatorname{Cl(G)}$ then recovers edge and clique approximate agreement on $G$, respectively.

We study the solvability of these tasks, that is, whether there exists an algorithm (a \emph{protocol}) that specifies, for each process, what operations to perform based on its local state, and that guarantees termination as well as the validity and agreement conditions in every execution.

The solvability of approximate agreement depends on the input space, the validity condition, and the number of processes.
Exact characterizations of solvability are not always known, but there is a general intuition: % that has appeared in various forms in the past
approximate agreement is solvable if and only if the space does not have ``holes'': $\epsilon$-approximate agreement and $m$-dimensional agreement are always solvable~\cite{dolev1986reaching,mendes2015multidimensional}, %and indeed the $d$-dimensional Euclidean space has no ``holes,''
 while clique agreement is solvable for all $n$ on trees and graphs whose clique graph is a tree, and for $n\leq 2$ on connected graphs~\cite{castaneda2018convergence,alcantara2019topology}.
Moreover, $\epsilon$-generalized $2$-dimensional $\epsilon$-agreement is claimed to be unsolvable with holes of radius $\epsilon$~\cite{herlihy1993asynchronous}, and edge and clique agreement are unsolvable for $n \geq 3$ on graphs with cycles of length $k \geq 4$~\cite{castaneda2018convergence}.

A ``hole'' is considered as an obstruction to solvability. However, there are actually different (related) ways in which a space can have a  ``hole''. In~\cite{herlihy1993asynchronous}, $\epsilon$-holes are $\epsilon$-isolated points. For graphs, an intuitive notion is that of $k$-cycles for $k > 3$. 
Several algebraic notions, discussed in Sections~\ref{section:complexes} and~\ref{section:combinatorial}, are studied in algebraic topology, and used in distributed computing. Those can be computable~\cite{Havlicek2000}, %too much citations already!,algebraicSpans},
for example via homology, or non-computable, like \emph{contractibility}, used in the conjecture of Ledent~\cite{Ledent21}: % considered the non-computable notion of contractibility in his conjecture.

\begin{conjecture}[Solvability of simplex agreement~\cite{Ledent21}]\label{conj:simplex}
Let $\mathcal{C}$ be a finite non-empty simplicial complex. Simplex agreement on $\mathcal C$ is solvable for all numbers $n$ of processes if and only if $\mathcal C$ is contractible.
\end{conjecture}

Conjecture~\ref{conj:simplex} exemplifies the advantages of such general problems: by abstracting away the specifics of each variant, general solvability criteria appear. Just as simplex agreement generalizes several discrete variants of approximate agreement, one of the goals of this paper is to provide a more general continuous approximate agreement task, not only on Euclidean spaces, but on all metric spaces satisfying ``convexity'' properties. Our second objective is to determine under which conditions simplex agreement is solvable.

\subsection*{Contributions and Plan} 
We first recall notions of distributed computing and topology in Sections~\ref{section:distributed} and Section~\ref{section:complexes}, respectively. 
In Section~\ref{section:solvability}, we focus on discrete spaces: graphs and simplicial complexes. Our main result is a levelwise refinement of Conjecture~\ref{conj:simplex} involving a finer topological property: $n$-connectivity (Definition~\ref{def:n-connectivity}). Intuitively, an $n$-connected space has no holes of dimension at most $n$.
%SR: removed to complie
\begin{mainthm}
%\begin{theorem} 
Simplex agreement on $\mathcal C$ is solvable for $n + 1$ processes in the wait-free asynchronous read-write model
if and only if $\mathcal C$ is $(n-1)$-connected.
\end{mainthm}
In particular, a space $X$ is $n$-connected for all $n \in \mathbb N$ if and only if it is contractible.
We then examine combinatorial criteria for determining the solvability of simplex agreement (Section~\ref{section:combinatorial}), and discuss the special case of clique agreement, for which Theorem~\ref{th:main} subsumes several existing combinatorial criteria, and covers strictly more examples (Section~\ref{section:application}).

%unlike our explicit protocol for approximate agreement on CUB spaces, the proof 
% Finally, we provide an example graph on which the existing combinatorial criteria do not apply, but Theorem~\ref{th:main} does. 

However, Theorem~\ref{th:main} relies on the Existence Theorem~\ref{th:carrier}, and thus does not provide an explicit protocol. To get explicit protocols, we turn our attention in Section~\ref{section:explicit} to continuous spaces. Our goal is to exploit the geometric nature of simplicial complexes to obtain explicit protocols solving simplex agreement. % Since geometric realizations of simplicial complexes are not   
To do so, we consider CUB spaces (Section~\ref{section:CUB}), a very general form of metric space with convexity properties, and introduce approximate agreement on CUB spaces, which subsumes classical $\epsilon$-agreement and $m$-dimensional approximate agreement. Our main results are that $\epsilon$-agreement on a CUB space is solvable for any number of processes by an explicit protocol, and that \emph{collapsible} simplicial complexes are CUB (Section~\ref{section:CUB_from}), yielding explicit protocols for solving approximate agreement on this restricted class of complexes (Section~\ref{section:recovering}). In particular, collapsibility implies contractibility, but the reverse is false.

\ifextended
\else
An extended version of this paper, including all the proofs that were omitted in this version, is available as a preprint~\cite{albert2026solving}.
\fi

\subsection*{Related work}

Approximate agreement belongs to a broader family of distributed agreement problems, rooted in the consensus task, in which processes must agree on a single value. The impossibility of consensus in asynchronous systems with crash failures~\cite{fischer1985impossibility} led to a number of weaker agreement tasks, among them approximate agreement. Other relaxations of consensus include $k$-set agreement~\cite{chaudhuri1993more}, in which the processes collectively decide at most $k$ distinct values; loop agreement~\cite{herlihy1997decidability}; and, more generally, rendezvous tasks~\cite{XingwuLiu2026,LiuXP09}. The $k$-set agreement task is an instance of simplex agreement, on the simplicial complex whose vertices are the possible values, and whose simplices are the sets of at most $k$ different values. Rendezvous tasks, on the other hand, are independent of simplex agreement, but their solvability also depends on connectivity criteria.

Other approximate agreement tasks on graphs have been proposed. Under shortest path validity~\cite{alistarh2023wait} (resp. minimal path validity~\cite{NowakR19}), outputs must lie on some shortest (resp. chordless) path between two inputs. Further criteria for the solvability or impossibility of approximate agreement on graphs have also been proposed. However, these criteria are mostly combinatorial and, a priori, not topological: clique approximate agreement is solvable for all nicely bridged or radius-one graphs under shortest path validity~\cite{alistarh2023wait}, the existence of an \emph{AER labeling}~\cite{alistarh2023wait} on a graph prevents solvability for $n\geq 3$ processes, and clique agreement is impossible on graphs satisfying the \emph{clique containment condition}~\cite{Liu22}. Conjecture~\ref{conj:simplex} and our Theorem~\ref{th:main} hint that solvability depends less on traditional graph features than on the underlying higher-dimensional topological structure. The results of Section~\ref{section:application} confirm this intuition by showing that these combinatorial criteria imply topological properties.

Another general discrete approximate agreement task has been proposed by Nowak and Rybicki~\cite{NowakR19}. It relies on abstract convexity spaces~\cite{kay1971axiomatic}, which provide a notion of convexity for \emph{finite} structures. This framework yields solvability results for several discrete structures. On graphs, it applies to chordal graphs, a subclass of nicely bridged graphs.

Our main characterization, Theorem~\ref{th:main}, shows that simplex agreement is undecidable even for three processes, since it is equivalent to deciding simple connectivity. It is known that general task solvability is undecidable in the wait-free case for three processes~\cite{GafniK99} (as is loop agreement~\cite{HerlihyR03}), and in essentially any non-trivial model of computation~\cite{herlihy1997decidability}.

Finally, after our preprint~\cite{albert2026solving} was posted online on 22 June 2026, Rybicki and Verbitsky~\cite{rybicki2026solvability,verbitsky2026topological} independently posted a similar proof of Theorem~\ref{th:main} on 23 June 2026. The same proof strategy generalizes to $\mathcal C$-convex simplex agreement tasks, where convexity is understood in the sense of convexity spaces~\cite{kay1971axiomatic}. This allows them to address a larger set of questions on graphs, for instance the problems of monophonic and geodesic agreement instead of just clique agreement. However, they do not address the problem of \emph{continuous} approximate agreement.

\section{Distributed computing models and solvability}\label{section:distributed}

The tasks we introduce are stated independently from any distributed computing model. We study their solvability in the standard~\cite{attiya2004distributed,herlihy2013distributed} crash-prone wait-free asynchronous read-write shared memory model: processes may crash, communicate by atomically reading from and writing to shared registers, without any timing assumptions, and each process must complete its task in a finite number of steps regardless of the speed or failure of any other process.

We describe our protocol for approximate agreement on CUB spaces in the equivalent~\cite{herlihy2013distributed} crash-prone wait-free layered immediate snapshot model: 
processes progress in synchronous rounds.
At each round, the process with id $i$ writes its current view into the $i$-th cell of a fresh memory array, then takes an atomic snapshot of the entire memory. Because writes happen concurrently, a process's snapshot only includes the writes that occurred before its own write.

Finally, we prove Theorem~\ref{th:main} in the \emph{colorless} crash-prone wait-free layered immediate snapshot model: in this variation of the previous models, process ids are ignored; processes exchange their view but do not track who sent what. 
This model is equivalent to the two previous models for \emph{colorless} tasks only~\cite{HerlihyRRS2017}, which is the case for simplex agreement.

\section{Abstract and geometric simplicial complexes}\label{section:complexes}

In this section, we introduce the combinatorial and topological notions used in this document, starting with abstract simplicial complexes. We simply call them simplicial complexes, as opposed to geometric simplicial complexes.

\begin{definition}
A \textbf{simplicial complex} is a set $\mathcal{C}$ of non-empty sets called \textbf{simplices}, such that if $\sigma \in \mathcal{C}$, then any subset $\emptyset \neq \tau \subseteq \sigma$, called a \textbf{face} of $\sigma$, is also in $\mathcal{C}$. The \textbf{dimension} of $\sigma \in \mathcal{C}$ is $|\sigma|-1$. $0$ and $1$-dimensional simplices are called \textbf{vertices} and \textbf{edges}, respectively. $\mathcal{C}$ is \textbf{finite-dimensional} if there is a uniform bound on the dimensions of its simplices.  A \textbf{simplicial subcomplex} of $\mathcal{C}$ is a subset of $\mathcal{C}$ that is also a simplicial complex.
\end{definition}

Simplicial complexes are partially ordered by inclusion.
In the following, we consider the following simplicial subcomplexes of a simplicial complex $\mathcal C$: for all $n \in \mathbb N$, the \emph{$n$-skeleton} $\operatorname{skel}^n \mathcal C$ is $\mathcal C$, with simplices of dimension greater than $n$ removed.
Simplicial complexes have geometric counterparts:

\begin{definition}
    A \textbf{geometric simplicial complex} $\mathcal K$ in $\mathbb{R}^N$, $N \in \mathbb N$, is a simplicial complex whose simplices are sets of linearly independent points in $\mathbb{R}^N$.  
    The convex hull of a simplex is called a \textbf{geometric simplex} of $\mathcal K$.
    We regard $\mathcal K$ as the subspace of $\mathbb R^N$ defined by the union of its geometric simplices.
    We call \textbf{interior} of a simplex $\sigma$ the subspace $\sigma$ with its faces removed. A \textbf{subdivision} of $\mathcal K$ is a geometric simplicial complex $\mathcal L$ that is homeomorphic to $\mathcal K$, denoted $\mathcal K \simeq \mathcal L$, and such that every simplex of $\mathcal L$ is included in some simplex of $\mathcal K$.
\end{definition}

%Any geometric simplicial complex determines a simplicial complex. In the other direction, 
Any \emph{finite} simplicial complex $\mathcal C$ can be turned into a geometric simplicial complex $|\mathcal C|$, its \emph{geometric realization}~\cite{munkres2025elements}, whose underlying simplicial complex is $\mathcal C$. Given a simplex $\sigma$ of $\mathcal C$, we denote by $|\sigma|$ the corresponding geometric simplex in $|\mathcal C|$. 
We can now lift topological notions from simplicial complexes to their geometric realizations:

%\todo{SR: contractible is one of the most important notion of the paper, "homotopoty equivalent" "weakly" undefined. At least sa:Intuitively, a contractible space is one that can be continuously shrunk to a point within that space. Fig 1 should be referenced; adding a little arrow or explanation of how the full triangle is deformed into its two edges. AA: added full definition. Hide weak contractibility in a proof, not needed otherwise.}
\begin{definition}
    A topological space $X$ is \textbf{contractible} if it is homotopy equivalent to a point, i.e., if there is a point $* \in X$ and a continuous map $H: X\times [0,1] \to X$ such that $H(x, 0) = x$ and $H(x, 1) = *$. 
\end{definition}

Intuitively, a contractible space is one that has no holes in a strong sense: it can be continuously shrunk to a point within that space, see Figure~\ref{fig:complex}. 

\begin{definition}\label{def:n-connectivity}
    A topological space $X$ is \textbf{$n$-connected} for $n \in \mathbb N$ if it is non-empty, path connected, and its first $n$ homotopy groups $\pi_k(X)$, $1 \leq k \leq n$, are trivial, i.e., any continuous map $S^k \to  X$ from the $k$-sphere extends to a continuous map $D^{k+1} \to X$ from the $(k+1)$-disk. We also say that $X$ is \textbf{$(-1)$-connected} if it is non-empty. 
\end{definition}

Intuitively, a $(n+1)$-connected spaces has no holes of dimensions less than $n$. For geometric simplicial complexes, connectedness is related to contractibility:

\begin{proposition}
$|\mathcal C|$ is $n$-connected for all $n \in \mathbb N$ if and only if it is contractible.
\end{proposition}
\begin{proof}
A space $X$ is $n$-connected for all $n \in \mathbb N$ if and only if it is weakly contractible. By the Whitehead theorem~\cite{whitehead1949combinatorial}, $|\mathcal C|$ is weakly contractible if and only if it is contractible. \qed
\end{proof}

We say that $\mathcal C$ is \emph{contractible} (resp. $n$-\emph{connected} for $n \in \mathbb N$) if its geometric realization $|\mathcal C|$ is. By definition, $\mathcal C$ is $n$-connected for all $n \in \mathbb N$ if and only if it is contractible. Contractibility has a combinatorial counterpart:

\begin{definition}
    Let $\mathcal C$ be a \emph{finite} simplicial complex. A \textbf{free face} $(\tau, \sigma)$ of $\mathcal C$ is a pair of simplices of $\mathcal C$ satisfying: $(1)$ $\sigma$ is maximal (for inclusion) in $\mathcal C$, $(2)$ $\tau$ is included but not equal to $\sigma$, and $(3)$  $\tau$ is not included in any other simplex of $\mathcal C$.
% \begin{itemize}[noitemsep]
%     \item $\sigma$ is maximal (for inclusion) in $\mathcal C$, 
%     \item $\tau$ is included but not equal to $\sigma$,
%     \item $\tau$ is not included in another simplex of $\mathcal C$.
% \end{itemize}
A \textbf{collapse} of $\mathcal C$ is the removal of all simplices $\tau \subseteq \gamma \subseteq \sigma$, where $(\tau, \sigma)$ is a free face. $\mathcal C$ is \textbf{collapsible} if there is a sequence of collapse leading to a point.
\end{definition}

A collapsible simplicial complex is contractible, but the converse is false, even for finite complexes, see for example the dunce hat %~\cite{zeeman1963dunce}
or the room with two holes\footnote{See, for example~\url{https://en.wikipedia.org/wiki/Collapse_(topology)}.}.

\begin{figure}[t]
    \centering
        \begin{tikzpicture}[scale=.41]

\fill[blue!12, draw=black, thick] (90:2) -- (210:2) -- (-30:2) -- cycle;

\fill[blue!20, draw=black, rounded corners=3mm] (90:1.5) -- (210:1.5) -- (-30:1.5) -- cycle;
\fill[blue!30, draw=black, rounded corners=3mm] (90:1) -- (210:1) -- (-30:1) -- cycle;
\fill[blue!40, draw=black]  (0,0) circle (0.25);
\filldraw[black] (0,0) circle (1pt);

%\draw[gray, thick, rounded corners=2mm] (0,0) -- (1,0) -- (1,1.732) -- cycle;

\end{tikzpicture}
        ~ ~ ~ ~ ~
        \begin{tikzpicture}[scale=.65]

% First: filled equilateral triangle with vertex dots
\fill[blue!12, draw=black, thick] (0,0) -- (2,0) -- (1,1.732) -- cycle;
\filldraw[black] (0,0) circle (2pt);
\filldraw[black] (2,0) circle (2pt);
\filldraw[black] (1,1.732) circle (2pt);

% Arrow between 1st and 2nd
\draw[->, thick] (2.4,0.866) -- (2.7,0.866);

% Second: same triangle without right side and no fill, with vertex dots
\begin{scope}[xshift=3cm]
\draw[thick] (0,0) -- (2,0);
\draw[thick] (0,0) -- (1,1.732);
% Right side omitted
\filldraw[black] (0,0) circle (2pt);
\filldraw[black] (2,0) circle (2pt);
\filldraw[black] (1,1.732) circle (2pt);
\end{scope}

% Arrow between 2nd and 3rd
\draw[->, thick] (5.4,0.866) -- (5.7,0.866);

% Third: just two connected vertices
\begin{scope}[xshift=6cm]
\draw[thick] (0,0) -- (1,1.732);
\filldraw[black] (0,0) circle (2pt);
\filldraw[black] (1,1.732) circle (2pt);
\end{scope}

% Arrow between 3rd and 4th
\draw[->, thick] (7.6,0.866) -- (7.9,0.866);

% Fourth: single vertex dot
\begin{scope}[xshift=9cm]
\filldraw[black] (0,0) circle (2pt);
\end{scope}

\end{tikzpicture}
    \caption{Contractibility and collapsibility of the standard $2$-simplex. Left: a homotopy contracting the triangle to its midpoint. Right: a sequence of elementary collapses reducing it to a single vertex.}
\end{figure}
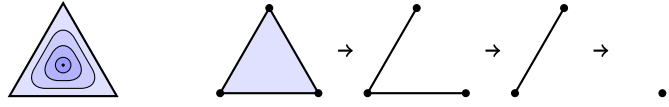

Aside from simplicial complexes, we will also encounter (geometric) \textbf{cubical complexes}, defined analogously to simplicial complexes, but using cubes instead of simplices. They consist of points, edges, filled squares, and more generally filled cubes in $\mathbb R^N$, glued along their faces.

Finally, geometric simplicial and cubical complexes are topological spaces, but also pseudometric spaces.
Each geometric simplex of a geometric simplicial complex $\mathcal K$ is convex in $\mathbb R^N$, hence can be equipped with the \textbf{flat equilateral} metric: the Euclidean metric where each edge has length $1$. This induces the \textbf{piecewise flat equilateral} pseudometric on $|\mathcal C|$: the length of a path in $|\mathcal C|$ is defined as the sum of the lengths of its intersections with the individual simplices, and the distance between two points $x,y$ is the infimum over all paths from $x$ to $y$ of their length, and $\infty$ if the two points are not connected by a path.
If $\mathcal K$ is path-connected, this pseudometric is a metric. 
Similarly, geometric cubical complexes are also equipped with the piecewise flat equilateral pseudometric, where each edge of a cube has length $1$.

A \textbf{piecewise-linear} homeomorphism $\mathcal K \to \mathcal L$, or \textbf{PL}-homeomorphism, between geometric simplicial or cubical complexes $\mathcal K$ and $\mathcal L$, is the data of two simplicial subdivisions $\mathcal K'$ and $\mathcal L'$ of $\mathcal K$ and $\mathcal L$, respectively, along with an homeomorphism $\mathcal K' \simeq \mathcal L'$, such that each geometric simplex of $\mathcal K'$ is mapped linearly to a geometric simplex of $\mathcal L'$ with respect to the piecewise flat equilateral metrics on $\mathcal K$ and $\mathcal L$. See Figure~\ref{fig:CUB} for an example of PL-homeomorphism between a simplicial and a cubical complex.

\section{The solvability of simplex agreement}\label{section:solvability}

Here we prove our main result for discrete spaces, a full characterization theorem of solvability of simplex agreement, in Section~\ref{sec:mainDisc}.
Then, in Section~\ref{section:combinatorial},
 we examine different criteria for determining when  a given  simplicial complex $\mathcal C$ satisfies the criteria of the theorem.
 Finally in Section~\ref{section:application} we show how the theorem implies results about solvability of clique agreement.

\subsection{The characterization}
\label{sec:mainDisc}
Conjecture~\ref{conj:simplex} concerns the solvability of simplex agreement for all numbers of processes. However, existing results already indicate that solvability depends on the number of processes: specifically, clique agreement under shortest path validity is always solvable for two processes~\cite{alistarh2023wait}. In this section, we restate simplex agreement within the framework of combinatorial topology~\cite{herlihy2013distributed} and apply Theorem~\ref{th:carrier} below to prove our main Theorem~\ref{th:main}. This theorem provides a finer characterization of the solvability of simplex agreement based on the number of processes, thereby generalizing previous observations to arbitrary dimensions.

\begin{definition}[{\cite{herlihy2013distributed}}]
    A \textbf{colorless task} is a triple $(\mathcal I, \mathcal O, \phi)$, where:
    \begin{itemize}
        \item $\mathcal I$ and $\mathcal O$ are simplicial complexes: the \textbf{input} and \textbf{output} complexes;
        \item $\phi : \mathcal I \to \subc{\mathcal O}$ is a \textbf{carrier map}, i.e., a function from simplices of $\mathcal I$ to simplicial subcomplexes of $\mathcal O$, monotonic in the sense that for simplices $\kappa \subset \sigma$ in $\mathcal I$, $\phi(\kappa)$ is a simplicial subcomplex of $\phi(\sigma)$.
    \end{itemize}
\end{definition}

Let us now fix a finite non-empty simplicial complex $\mathcal C$. The colorless task associated to simplex agreement on $\mathcal C$ is the triple $(\Delta, \mathcal C, \phi)$ with:
\begin{itemize}
    \item $\Delta = \{\sigma \subseteq \mathcal{C}_0 \}$ the \emph{full} simplicial complex on the vertices $\mathcal{C}_0$ of $\mathcal C$.
    \item $\phi$ is the carrier map defined on subsets of vertices $\sigma$ of $\mathcal C$ by $\phi(\sigma) = \sigma$ regarded as a simplicial subcomplex of $\mathcal C$ if $\sigma \in \mathcal C$, and $\phi(\sigma) = \mathcal{ C}$ otherwise.
\end{itemize}

This colorless task is a special case of \textbf{convergence tasks}~\cite{herlihy2013distributed}, i.e. triples $(\mathcal I, \mathcal C, \phi)$ where $\mathcal C$ is any simplicial complex, $\mathcal I$ a simplicial complex over the vertices of $\mathcal C$, and $\phi(\sigma) = \sigma$ when $\sigma \in \mathcal C$, or $\mathcal C$ otherwise. From this perspective, simplex agreement is the \emph{canonical} {convergence task}. 
The solvability of colorless tasks in the wait-free asynchronous read-write model is characterized by the existence of certain continuous \emph{carried} maps~\cite{HerlihyRRS2017}:
a continuous map $f : |\mathcal K| \to |\mathcal L|$ is \textbf{carried by} a carrier map $\phi$ if for all geometric simplices $\sigma \in |\mathcal K|$, we have $f(\sigma) \subseteq |\phi(\sigma)|$.

\begin{theorem}[\cite{HerlihyRRS2017}]\label{th:carrier}
A colorless task \( (\mathcal{I}, \mathcal{O}, \phi) \) has a wait-free \( (n + 1) \)-process layered protocol if and only if there is a continuous map
$f : |\operatorname{skel}^n \mathcal{I}| \to |\mathcal{O}|$ carried by \( \phi \). 
\end{theorem}

 In our situation, the application of Theorem~\ref{th:carrier} is straightforward, and the carrying condition means that
for $n \in \mathbb N$ and all simplices $\sigma$ in $\operatorname{skel}^n\Delta$, if $\sigma \in \mathcal C$ then $f(|\sigma|) \subseteq |\sigma|$, otherwise $f(|\sigma|) \subseteq |\mathcal C|$. We deduce the following lemma\ifextended\else, proved in the extended version~\cite{albert2026solving}\fi.

\begin{lemma}\label{lemma:carried}
Fix $n \in \mathbb N$.
There is a continuous map $f: |\operatorname{skel}^n\Delta| \to |\mathcal C|$ carried by $\phi$ if and only if the inclusion $i : |\operatorname{skel}^n \mathcal C| \hookrightarrow |\mathcal C|$ can be extended to a continuous map $|\operatorname{skel}^n\Delta| \to |\mathcal C|$.
\end{lemma}
\ifextended
\begin{proof}
For the forward direction, let $f : |\operatorname{skel}^n\Delta| \to |\mathcal C|$ be a continuous map carried by $\phi$. Let us build an extension of $i$. For every simplices $\sigma \in \operatorname{skel}^n \mathcal C$, both $f(x)$ and $x$ lie in the convex subspace $|\sigma|$, so we can consider $H_\sigma: |\sigma| \times [0,1] \to |\sigma|$ the homotopy defined by $H_\sigma(x,t) = (1-t)f(x) +tx$. For all simplices $\sigma, \kappa \in  \operatorname{skel}^n \mathcal C$, $H_\sigma$ and $H_\kappa$ agree on $|\sigma| \times [0,1]$ and $|\kappa| \times [0,1]$. By the pasting lemma, they assemble into an homotopy $H: |\operatorname{skel}^n \mathcal C| \times [0,1] \to |C|$ from $f|_{|\operatorname{skel}^n \mathcal C|}$ to $i$. Since the pair $(|\operatorname{skel}^n \Delta|, |\operatorname{skel}^n \mathcal C|)$ has the homotopy extension property by~\cite[Proposition 0.16]{hatcher}, and because $f$ and $H$ agree on ${\operatorname{skel}^n \mathcal C} \times \{0\}$, there is an homotopy $H': |{\operatorname{skel}^n \Delta}| \times [0,1] \to |\mathcal C|$ from $f$ to an extension of $i$.

For the reverse direction, if an extension $f$ exists, then it is automatically carried: for all simplices $\sigma \in \operatorname{skel}^n\Delta$, if $\sigma \in \operatorname{skel}^n \mathcal C$ then $f(|\sigma|) = |\sigma|$, otherwise $f(|\sigma|) \subseteq |\mathcal C|$. \qed
\end{proof}
\fi

\begin{theorem}\label{th:main}
For all $n \in \mathbb N$, there is an $(n+1)$-process layered protocol for simplex agreement on $\mathcal C$ if and only if $\mathcal C$ is $(n-1)$-connected.
\end{theorem}
\ifextended
\begin{proof}
By Theorem~\ref{th:carrier} and Lemma~\ref{lemma:carried}, it suffices to show that $ \mathcal C$ is $(n-1)$-connected if and only if the inclusion $i : |\operatorname{skel}^n \mathcal C| \hookrightarrow |\mathcal C|$ extends to a continuous map $|\operatorname{skel}^n \Delta| \to |\mathcal C|$. 

For the forward direction, suppose $\mathcal C$ is $(n-1)$-connected. Let us build a continuous extension $g_k: |\operatorname{skel}^k\Delta| \to |C|$ of $i_k : |\operatorname{skel}^k \mathcal C| \hookrightarrow |\mathcal C|$ for all $0 \leq k \leq n$ by induction on $k$:
    \begin{itemize}
        \item For $k = 0$, $\operatorname{skel}^0\Delta = \operatorname{skel}^0 \mathcal C$ so we define $g_0 = i_0$.
        \item Suppose $g_k$ has been defined for $0 \leq k < n$. For all $(k+1)$-simplices $\sigma$ in $\operatorname{skel}^{k+1}\Delta$, either $\sigma \in \operatorname{skel}^{k+1} \mathcal C$, and we set $g_\sigma : |\sigma| \to |\mathcal C|$ to be the inclusion. In the other case, the boundary $|\partial \sigma|$ is a $k$-sphere in $\operatorname{skel}^{k}\Delta$, on which $g_k$ is defined. But since $|\mathcal C|$ is $(n-1)$-connected, $\pi_k(|\mathcal C|) = 0$, so $g_k|_{|\partial \sigma|}: |\partial \sigma| \to |\mathcal C|$ is null homotopic, and we can extend $g_k|_{|\partial \sigma|}$ to a continuous map $g_\sigma: {| \sigma|} \to |\mathcal C|$. Again by the pasting lemma, this defines a continuous map $g_{k+1}: |\operatorname{skel}^{k+1}\Delta| \to |C|$.
    \end{itemize}

    For the reverse direction, given an extension $g: |\operatorname{skel}^n\Delta| \to |\mathcal C|$ of the inclusion $i : |\operatorname{skel}^n \mathcal C| \hookrightarrow | \mathcal C|$, let us show for all $0 \leq k \leq n-1$ that all the continuous maps $S^k \to |\mathcal C|$ are null‑homotopic, i.e., that $\pi_k(|\mathcal C|)=0$.
    Let $s: S^k \to |\mathcal C|$ be such a map. By the cellular approximation theorem, $s$ is homotopic to a cellular map $s': S^k \to |\mathcal C|$, whose image lies in $|\operatorname{skel}_k \mathcal C|$. Without loss of generality, we can work with $s'$ instead.
    Consider the map $s'': S^k \xrightarrow{s'} |\operatorname{skel}^kC| \hookrightarrow |\operatorname{skel}^n\Delta|$. Since $g$ is an extension, we have $g \circ s'' = s' : S^k \to |\mathcal C|$.
    In addition, since $\Delta$ is contractible, $|\operatorname{skel}^n\Delta|$ is $n-1$-connected, again by the cellular approximation theorem. Hence $s''$ is null-homotopic, and so is $ g \circ s'' = s'$.  \qed
\end{proof}
\else
\begin{proof}
See~\cite{albert2026solving} for the detailed proof. 
By Theorem~\ref{th:carrier} and Lemma~\ref{lemma:carried}, it suffices to show that $ \mathcal C$ is $(n-1)$-connected if and only if the inclusion $i : |\operatorname{skel}^n \mathcal C| \hookrightarrow |\mathcal C|$ extends to a continuous map $|\operatorname{skel}^n \Delta| \to |\mathcal C|$. 
For the forward direction, we suppose that $\mathcal C$ is $(n-1)$-connected, and we build continuous extensions $g_k: |\operatorname{skel}^k\Delta| \to |C|$ of $i_k : |\operatorname{skel}^k \mathcal C| \hookrightarrow |\mathcal C|$ for all $0 \leq k \leq n$ by induction on $k$.
For the reverse direction, given an extension $g: |\operatorname{skel}^n\Delta| \to |\mathcal C|$ of the inclusion $i : |\operatorname{skel}^n \mathcal C| \hookrightarrow | \mathcal C|$, we show for all $0 \leq k \leq n-1$ that the continuous maps $S^k \to |\mathcal C|$ are null‑homotopic. \qed
\end{proof}
\fi

Notice that in the case $n=0$, since we supposed that $\mathcal C$ was non-empty, it is by definition $(-1)$-connected, and a single process can indeed always solve simplex agreement on $\mathcal C$.
We easily recover Conjecture~\ref{conj:simplex}: simplex agreement on $\mathcal C$ is solvable for all numbers of agents if and only if $\mathcal C$ is $n$-connected for all $n \in \mathbb N$, if and only if $\mathcal C$ is contractible.

\subsection{Combinatorial criteria for connectivity}\label{section:combinatorial}

In this section, we examine different criteria for determining for which values of $n \in \mathbb N$ a given finite simplicial complex $\mathcal C$ is $n$-connected. First observe that it is trivial to decide whether $\mathcal C$ is non-empty and path-connected. However, $1$-connectivity, also called \textbf{simple-connectivity}, is already undecidable, as it reduces to the word problem~\cite{rabin1956recursive}. By the Hurewicz theorem, we see that this is the main obstacle to decidability: 
\begin{itemize}
    \item Suppose $\mathcal C$ is non-empty, path-connected, and simply connected. Then, for all $n > 1$, the $n$th homotopy group $\pi_n(|\mathcal C|)$ is isomorphic to the $n$th reduced singular homology group with integer coefficients $\tilde H_n(|\mathcal C|)$, which is  isomorphic to the simplicial homology group $\tilde H_n(\mathcal C)$, and the latter is efficiently computable~\cite{dumas2003computing}. In addition, $\tilde H_0(\mathcal C)$ is trivial. Therefore, $\mathcal C$ is $n$-connected for $n > 1$ if and only if its first $n+1$ reduced homology groups $\tilde H_k(\mathcal C)$, $0 \leq k \leq n$, are trivial. Since $\mathcal C$ is finite of maximal dimension $d$, it suffices to compute its first $d$ homology groups (the higher ones are trivial).
    \item If $\mathcal C$ is non-empty and path-connected (but not necessarily simply connected), then homology only gives an upper bound: if the $n$th reduced homology groups $\tilde H_n(\mathcal C)$ is not trivial, then $\mathcal C$ is not $n$-connected. 
\end{itemize}

Nevertheless, there are families of finite simplicial complexes for which simple-connectivity is decidable. For instance, \emph{collapsibility}, which is decidable on finite simplicial complexes, implies contractibility hence simple-connectivity.
More interestingly, \emph{shellability} is $\mathrm{NP}$-complete~\cite{goaoc2019shellability}, hence decidable, and implies useful topological properties:

%\todo{SR: take out the "pure" from this definition. Put it where it is used. Also take out the intuition.}
\begin{definition}[\cite{bjorner1996shellable,bruggesser1971shellable}]
 A finite simplicial complex is \textbf{shellable} if there is an ordering $\sigma_1, \sigma_2, \ldots $ of its maximal simplices such that for all $k \geq 2$, the simplicial subcomplex: $$ {\Big (}\bigcup _{i=1}^{k-1}\sigma_{i}{\Big )}\cap \sigma_{k}$$ is a $(\dim \sigma_k - 1)$-dimensional complex such that all simplices are faces of some $(\dim \sigma_k - 1)$-dimensional simplex.
 \end{definition}

Intuitively, a shellable complex can be obtained from a well-behaved gluing of its maximal faces.

\begin{theorem}[\cite{bjorner1996shellable}]
A shellable simplicial complex has the
homotopy type of a wedge of spheres, where for each
dimension $i$, the number of $i$-spheres is the number of $i$-dimensional maximal simplices $\sigma_k$ whose entire boundary is contained in the union of the earlier maximal simplices $\bigcup _{i=1}^{k-1}\sigma_{i}$.
\end{theorem}

A wedge of spheres is $n$-connected exactly when every sphere in the wedge has dimension strictly greater than $n$, hence the connectivity of shellable complexes is decidable. This provides examples of simplicial complexes on which simplex agreement is solvable up to any fixed number $n \in \mathbb N$ of processes: any finite triangulation $\mathcal C$ of the $(n+1)$-sphere $S^{n+1}$ is $n$-connected, but not $n+1$-connected, hence simplex agreement on $\mathcal C$ is solvable for $k$ processes if and only if $k \leq n+2$. 

\subsection{Application: clique approximate agreement}\label{section:application}

The connectivity criterion of Theorem~\ref{th:main} recovers prior results for clique agreement under clique validity. Concretely, previously proposed combinatorial criteria translate into topological ones. Some of these observations were already made in~\cite{Ledent21}:
\begin{itemize}
    \item the clique complex of a $k$-cycle is the contractible standard $2$-simplex for $k=3$, but is not simply-connected for $k > 3$;
    \item the clique complex of a graph whose clique graph is a tree is contractible;
    \item the clique complex of any connected graph is also path connected;
    \item the clique complex of a radius-one graph is a cone, hence contractible;
    \item the clique complex of a chordal graph is contractible (see e.g.~\cite{adiprasito2016higher});
     \item the clique complexes of weakly bridged graphs are exactly the weakly systolic complexes, which are contractible, see~\cite{brevsar2013bucolic,osajda2013combinatorial}.
     \ifextended\else
    \item the existence of an \emph{AER labeling} of~\cite{alistarh2023wait} on a graph $G$ prevents simple-connectivity of $\operatorname{Cl}(G)$, see the extended version~\cite{albert2026solving};
    \item the $k$-clique containment condition on a graph $G$ from~\cite{Liu22} prevents $(k-1)$-connectivity of $\operatorname{Cl}(G)$, see~\cite{albert2026solving}. However, the relation between the connectivity of the clique complex $\operatorname{Cl}(G)$ and the chromatic number of subgraphs $A$ of $G$ is less clear.
    \fi
\end{itemize}

\ifextended
In addition, we show that the existence of an AER labeling~\cite{alistarh2023wait} on a graph $G$ prevents its clique complex $\operatorname{Cl}(G)$ from being simply connected, which recovers in particular that the existence of an AER labeling is an impossibility criterion for clique agreement.

\begin{definition}[\cite{alistarh2023wait}]
    An \textbf{AER labeling} of a graph $G$ is a labeling of the vertices of \(G\) by \(\{0,1,2\}\), such that:
    \begin{itemize}
        \item \(G\) contains no triangle with three different labels;
        \item \(G\) contains a cycle \(C\) with exactly one vertex labeled \(1\) and its two neighbors in \(C\) labeled \(0\) and \(2\).
    \end{itemize}
\end{definition}

\begin{theorem}
If a graph $G$ admits an AER labeling, then its clique complex $\operatorname{Cl}(G)$ is not simply-connected.
\end{theorem}
\begin{proof}
Consider the standard $2$-simplex $\Delta_2$ with vertices $\{0,1,2\}$.
Define a simplicial map \(f\colon \operatorname{Cl}(G)\to \Delta_2\) by sending each vertex to its label.  
The first condition of the AER labeling implies the image lies in the \(1\)-skeleton \(\partial\Delta_2\cong S^1\), giving a continuous map
\[
f: |\operatorname{Cl}(G)|\to S^1 .
\]
The cycle \(C\) from the second condition maps to a closed edge‑path in \(S^1\): from \(1\) to \(0\), then from \(0\) to \(2\) using only vertices labeled \(0,2\), and finally back to \(1\).  
Every edge between \(0\) and \(2\) uses the unique edge \(02\); a path from \(0\) to \(2\) using only \(0,2\) must traverse that edge an odd number of times. Hence \(f(C)\) is homotopic to the loop \(1\to0\to2\to1\), which winds once around \(S^1\).  
Consequently \(f_*\colon \pi_1(|\operatorname{Cl}(G)|)\to \pi_1(S^1)\) is non‑trivial.  
Therefore \(\pi_1(|\operatorname{Cl}(G)|)\) is non‑trivial, and \(\operatorname{Cl}(G)\) is not simply‑connected. \qed
\end{proof}

Similarly, we show that the clique containment condition, introduced in~\cite{Liu22} as an impossibility criterion for the clique agreement task, implies a fine connectivity condition, which implies an upper bound on the maximal number of processes for which clique agreement is solvable. The relation between this bound and the chromatic number bound given in~\cite{Liu22} is however not clear.  

\begin{definition}[\cite{Liu22}]
Fix $k \geq 2$.
A graph $G$ satisfies the $k$-clique containment conditions if there is a subgraph $A$ of $G$ such that:
\begin{enumerate}
    \item every $(k - 1)$-clique in $A$ is contained in exactly two $k$-cliques in $A$;\label{cond1}
    \item there is a $k$-clique in $A$ that is not contained in any $(k+1)$-clique in $G$.\label{cond2}
\end{enumerate}

\end{definition}

\begin{theorem}
If a graph $G$ satisfies the $k$-clique containment, then its clique complex $\operatorname{Cl}(G)$ is not $(k-1)$-connected.
\end{theorem}
\begin{proof}
In the following, we consider the homology of $\operatorname{Cl}(G)$ with coefficients in  $\mathbb{Z}/2\mathbb{Z}$.
Let $S$ be the sum of all $(k-1)$-simplices that correspond to $k$-cliques of the subgraph $A$.
By the first condition, because we work mod~$2$, adding a face twice yields zero, hence $\partial S = 0$ and $S$ is a $(k-1)$-cycle in $\operatorname{Cl}(G)$.
Now, the second condition guarantees the existence of a $(k-1)$-simplex $\sigma$ that is not a face of any $k$-simplex of $\operatorname{Cl}(G)$.
Therefore, if $S$ were the boundary of some $k$-chain $S'$, then the coefficient of $\sigma$ in $\partial S'$ would be $0$. But $\sigma$ appears with coefficient $1$ in $S$, hence $S$ represents a non‑trivial class in $H_{k-1}\bigl(\operatorname{Cl}(G); \mathbb{Z}/2\mathbb{Z}\bigr)$.
But if $\operatorname{Cl}(G)$ was $(k-1)$-connected, then we would have by the universal coefficient theorem and the Hurewicz theorem, $H_{k-1}\bigl(\operatorname{Cl}(G); \mathbb{Z}/2\mathbb{Z}\bigr) = 0$. \qed
\end{proof}
\fi

We can also provide graphs for which simplex agreement is solvable up to a given number $n \in \mathbb N$ of processes: consider the barycentric subdivision $\operatorname{bary}(\mathcal C)$ of any finite triangulation $\mathcal C$ of $S^{n+1}$. Since $|\mathcal C| \simeq |\operatorname{bary}(\mathcal C)|$, $\operatorname{bary}(\mathcal C)$ has the same topological properties as $\mathcal C$. In addition, $\operatorname{bary}(\mathcal C)$ is \textbf{flag}, i.e., is the clique complex of its $1$-skeleton. Therefore, by Theorem~\ref{th:main}, clique agreement on $G$ is solvable for $k$ processes if and only if $k \leq n+1$.% (remark that in this situation, the clique containment condition~\cite{Liu22} applies). 

More interestingly, we can construct graphs for which Theorem~\ref{th:main} applies, but no existing combinatorial criterion does. Consider the simplicial complex $\mathcal C$ in Figure~\ref{fig:complex}. $\mathcal C$ is again flag, i.e., is the clique complex of its $1$-skeleton. Its geometric realization is homotopy equivalent to the $2$-sphere $S^2$. Hence, $\mathcal C$ is $1$-connected but not $2$-connected. Theorem~\ref{th:main} applies, but since $\mathcal C$ is $1$-connected, the only existing combinatorial criterion that could apply is the clique containment condition~\cite{Liu22}, which does not, \ifextended as shown below\else see the extended version~\cite{albert2026solving}\fi. 

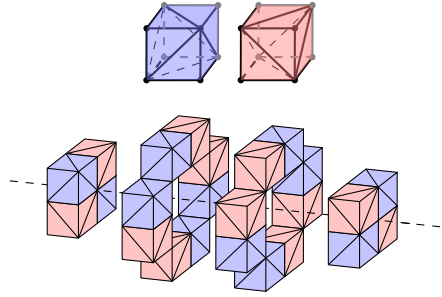
\begin{figure}[t]
\centering
    \begin{tikzpicture}[
    scale=.5,
    rotate around x={3},
    rotate around y={3},
    rotate around z={-1.3},
    % rotate around x={3},
    % rotate around y={3},
    % rotate around z={-1.3},
    % Define variable for side length
    declare function={L=1.4;} % Side length of the cube
]

% Define the coordinates of the cube vertices
\coordinate (A) at (0,0,0);       % front-bottom-left
\coordinate (B) at (L,0,0);       % front-bottom-right
\coordinate (C) at (L,L,0);       % front-top-right
\coordinate (D) at (0,L,0);       % front-top-left
\coordinate (E) at (0,0,L);       % back-bottom-left
\coordinate (F) at (L,0,L);       % back-bottom-right
\coordinate (G) at (L,L,L);       % back-top-right
\coordinate (H) at (0,L,L);       % back-top-left

% Draw edges
\draw[thick] (E) -- (F);
\draw[thick] (F) -- (G);
\draw[thick] (G) -- (H);
\draw[thick] (H) -- (E);

\draw[dashed, thick] (E) -- (A);
\draw[thick] (F) -- (B);
\draw[thick] (G) -- (C);
\draw[thick] (H) -- (D);

\draw[gray, dashed, thick] (A) -- (B);
\draw[gray, thick] (B) -- (C);
\draw[gray, thick] (C) -- (D);
\draw[gray, dashed, thick] (D) -- (A);

\draw[thick] (D) -- (G);
\draw[dashed, thin] (E) -- (B);
\draw[dashed, thin] (E) -- (D);
\draw[thick] (E) -- (G);
\draw[thick] (B) -- (G);
\draw[dashed, thin] (D) -- (B);

% Add vertices (2pt circles at each point)
\filldraw[gray] (A) circle (2pt);
\filldraw[gray] (B) circle (2pt);
\filldraw[gray] (C) circle (2pt);
\filldraw[gray] (D) circle (2pt);
\filldraw[black] (E) circle (2pt);
\filldraw[black] (F) circle (2pt);
\filldraw[black] (G) circle (2pt);
\filldraw[black] (H) circle (2pt);

\fill[blue!50, opacity=.45] (H) -- (G) -- (F) -- (E);

\fill[blue!50, opacity=.3] (H) -- (G) -- (C) -- (D);

\fill[blue!50, opacity=.3] (G) -- (C) -- (B) -- (F);

\fill[blue!50, opacity=.3] (H) -- (D) -- (C) -- (B) -- (F) -- (G) -- (H);

% \node[below left] at (A) {$A$};
% \node[below right] at (B) {$B$};
% \node[above right] at (C) {$C$};
% \node[above left] at (D) {$D$};
% \node[below left] at (E) {$E$};
% \node[below right] at (F) {$F$};
% \node[above right] at (G) {$G$};
% \node[above left] at (H) {$H$};

\end{tikzpicture} \begin{tikzpicture}[
    scale=.5,
    rotate around x={3},
    rotate around y={3},
    rotate around z={-1.3},
    % rotate around x={3},
    % rotate around y={3},
    % rotate around z={-1.3},
    % Define variable for side length
    declare function={L=1.4;} % Side length of the cube
]

% Define the coordinates of the cube vertices
\coordinate (A) at (0,0,0);       % front-bottom-left
\coordinate (B) at (L,0,0);       % front-bottom-right
\coordinate (C) at (L,L,0);       % front-top-right
\coordinate (D) at (0,L,0);       % front-top-left
\coordinate (E) at (0,0,L);       % back-bottom-left
\coordinate (F) at (L,0,L);       % back-bottom-right
\coordinate (G) at (L,L,L);       % back-top-right
\coordinate (H) at (0,L,L);       % back-top-left

% Draw edges
\draw[thick] (E) -- (F);
\draw[thick] (F) -- (G);
\draw[thick] (G) -- (H);
\draw[thick] (H) -- (E);

\draw[dashed, thick] (E) -- (A);
\draw[thick] (F) -- (B);
\draw[thick] (G) -- (C);
\draw[thick] (H) -- (D);

\draw[gray, dashed, thick] (A) -- (B);
\draw[gray, thick] (B) -- (C);
\draw[gray, thick] (C) -- (D);
\draw[gray, dashed, thick] (D) -- (A);

\draw[thick] (H) -- (C);
\draw[dashed, thin] (A) -- (F);
\draw[dashed, thin] (H) -- (A);
\draw[thick] (H) -- (F);
\draw[thick] (C) -- (F);
\draw[dashed, thin] (A) -- (C);

% Add vertices (2pt circles at each point)
\filldraw[gray] (A) circle (2pt);
\filldraw[gray] (B) circle (2pt);
\filldraw[gray] (C) circle (2pt);
\filldraw[gray] (D) circle (2pt);
\filldraw[black] (E) circle (2pt);
\filldraw[black] (F) circle (2pt);
\filldraw[black] (G) circle (2pt);
\filldraw[black] (H) circle (2pt);

\fill[red!50, opacity=.45] (H) -- (G) -- (F) -- (E);

\fill[red!50, opacity=.3] (H) -- (G) -- (C) -- (D);

\fill[red!50, opacity=.3] (G) -- (C) -- (B) -- (F);

\fill[red!50, opacity=.3] (H) -- (D) -- (C) -- (B) -- (F) -- (G) -- (H);

% \node[below left] at (A) {$A$};
% \node[below right] at (B) {$B$};
% \node[above right] at (C) {$C$};
% \node[above left] at (D) {$D$};
% \node[below left] at (E) {$E$};
% \node[below right] at (F) {$F$};
% \node[above right] at (G) {$G$};
% \node[above left] at (H) {$H$};

\end{tikzpicture} 
    \ \\
    \begin{tikzpicture}[
    scale=.45,
    rotate around x={3},
    rotate around y={-10},
    rotate around z={-1.3},
    declare function={L=1;} % Side length of the cube
]

% Define a macro for drawing one cube at a given position (x,y)
\newcommand{\redcube}[3]{% #1 = x offset, #2 = y offset, #3 = z offset (z is usually 0)
    \begin{scope}[shift={(#1*L, #2*L, #3*L)}]
        % Define coordinates relative to this cube's origin
        \coordinate (A) at (0,0,0);
        \coordinate (B) at (L,0,0);
        \coordinate (C) at (L,L,0);
        \coordinate (D) at (0,L,0);
        \coordinate (E) at (0,0,L);
        \coordinate (F) at (L,0,L);
        \coordinate (G) at (L,L,L);
        \coordinate (H) at (0,L,L);
        
        % Fill and draw the cube
        \fill[red!23] (H) -- (D) -- (C) -- (B) -- (F) -- (E) -- (H);
        \draw (E) -- (F) -- (G) -- (H) -- (E);
        \draw (F) -- (B) -- (C) -- (G);
        \draw (H) -- (D) -- (C);
        \draw (H) -- (F);
        \draw (C) -- (F);
        \draw (H) -- (C);
    \end{scope}
}

\newcommand{\bluecube}[3]{% #1 = x offset, #2 = y offset, #3 = z offset (z is usually 0)
    \begin{scope}[shift={(#1*L, #2*L, #3*L)}]
        % Define coordinates relative to this cube's origin
        \coordinate (A) at (0,0,0);
        \coordinate (B) at (L,0,0);
        \coordinate (C) at (L,L,0);
        \coordinate (D) at (0,L,0);
        \coordinate (E) at (0,0,L);
        \coordinate (F) at (L,0,L);
        \coordinate (G) at (L,L,L);
        \coordinate (H) at (0,L,L);
        
        % Fill and draw the cube
        \fill[blue!23] (H) -- (D) -- (C) -- (B) -- (F) -- (E) -- (H);
        \draw (E) -- (F) -- (G) -- (H) -- (E);
        \draw (F) -- (B) -- (C) -- (G);
        \draw (H) -- (D) -- (C);
        \draw (G) -- (E);
        \draw (G) -- (B);
        \draw (G) -- (D);
    \end{scope}
}

\draw[dashed] (-3,1,0) -- (0,1,0);

% ========= X = 0 =========
% _ r b _
% r b r b
% b r b r
% _ b r _
%\redcube  {0}{1}{0}
%\bluecube {0}{2}{0}

\%redcube  {0}{0}{1}
\bluecube {0}{1}{1}
\redcube  {0}{2}{1}
%\bluecube {0}{3}{1}

%\bluecube {0}{0}{2}
\redcube  {0}{1}{2}
\bluecube {0}{2}{2}
%\redcube  {0}{3}{2}

%\bluecube {0}{1}{3}
%\redcube  {0}{2}{3}

\draw[dashed] (0,1,0) -- (4,1,0);

% ========= X = 2 =========
% _ b r _
% b _ _ r
% r _ _ b
% _ r b _
\bluecube {3}{1}{0}
\redcube  {3}{2}{0}

\bluecube {3}{0}{1}
\redcube  {3}{3}{1}

\redcube  {3}{0}{2}
\bluecube {3}{3}{2}

\redcube  {3}{1}{3}
\bluecube {3}{2}{3}

\draw[dashed] (2.5,1,0) -- (7,1,0);

% ========= X = 4 (hole in middle) =========
% _ r b _
% r _ _ b
% b _ _ r
% _ b r _
\redcube  {6}{1}{0}
\bluecube {6}{2}{0}

\redcube  {6}{0}{1}
\bluecube {6}{3}{1}

\bluecube {6}{0}{2}
\redcube  {6}{3}{2}

\bluecube {6}{1}{3}
\redcube  {6}{2}{3}

\draw[dashed] (5.5,1,0) -- (10,1,0);

% ========= X = 6 (no hole, full chess) =========
% _ b r _
% b r b b
% r b r r
% _ r b _
%\bluecube {9}{1}{0}
%\redcube  {9}{2}{0}

%\bluecube {9}{0}{1}
\redcube  {9}{1}{1}
\bluecube {9}{2}{1}
%\bluecube {9}{3}{1}

%\redcube  {9}{0}{2}
\bluecube {9}{1}{2}
\redcube  {9}{2}{2}
%\redcube  {9}{3}{2}

%\redcube  {9}{1}{3}
%\bluecube {9}{2}{3}

\draw[dashed] (9,1,0) -- (10.8,1,0);

\end{tikzpicture}
    \caption{Two geometric simplicial complexes $A$ and $B$ in $\mathbb{R}^3$, together with a larger complex obtained by gluing copies of $A$ and $B$ along their boundary simplices, shown in an exploded view.}
    \label{fig:complex}
\end{figure}

\ifextended
\begin{proposition}
On the simplicial complex $\mathcal C$ of Figure~\ref{fig:complex}, the $k$-clique containment condition does holds for any $k \geq 2$:
\end{proposition}
\begin{proof} We distinguish cases depending on the value of $k$:
    \begin{itemize}
    \item For $k \in \{2, 3\}$, all edges (resp. triangles) are contained in at least one triangle (resp. $3$-simplex), so condition~\ref{cond2} cannot hold.
    \item For $k = 4$, Every $3$-simplex of $\mathcal C$ has at least one $2$-simplex on the boundary of $\mathcal C$, and so this $2$-simplex can only be the boundary of one $3$-simplex. So, if condition~\ref{cond2} holds, $A$ must contain a $3$-simplex, and condition~\ref{cond1} cannot hold.
    \item For $k \geq 5$, there is no $(k-1)$-simplices in $\mathcal C$, so condition~\ref{cond2} cannot hold. \qed
\end{itemize}
\end{proof}
\fi

\section{An explicit protocol for simplex agreement}\label{section:explicit}

We proved our solvability Theorem~\ref{th:main} via Theorem~\ref{th:carrier}, which does not provide an explicit protocol. The goal of this section is to obtain explicit, and perhaps intuitive, protocols for the restricted class of \emph{contractible} simplicial complexes. 

Let us now fix $\mathcal C$, a collapsible, finite, non-empty, path-connected simplicial complex.
Using the piecewise flat equilateral metric on $|\mathcal C|$, we would like to use the same iterated barycenter technique as in the protocols for continuous approximate agreement~\cite{dolev1986reaching,mendes2015multidimensional} to solve approximate agreement. However, $|\mathcal C|$ is not Euclidean, hence the notions of convexity and barycenter are ill-defined. Nevertheless, we show in the following the we can use a weaker notion of convex metric spaces: CUB spaces~\cite{roller2016poc,haettel2025link}(Definition~\ref{def:bicombing}).

We show in Section~\ref{section:CUB} how to solve approximate agreement on CUB spaces with an explicit midpoint protocol,  how to get a CUB space from a contractible simplicial complex in Section~\ref{section:CUB_from}, and how to deduce an explicit protocol for some \emph{subdivisions} of these complexes in Section~\ref{section:recovering}:

\begin{theorem}\label{th:snd}
There exists a finite subdivision $\mathcal K$ of $|\mathcal C|$ and an explicit protocol solving simplex agreement on $\mathcal K$.
%is solvable in the wait-free layered immediate snapshot model.
\end{theorem}
%However, our reduction does not show that $|\mathcal C|$ is CUB, and 
Without relying on a subdivision of $|\mathcal C|$, we only obtain solvability under \emph{weak validity}: if the inputs are identical, then all the outputs must equal that input.

\begin{corollary}
\label{cor:weaker}
There is an explicit protocol solving simplex agreement on $\mathcal C$ under weak validity.
\end{corollary}

\subsection{Approximate agreement on CUB spaces}\label{section:CUB}

In this section, we generalize the technique used for solving agreement problems on metric spaces~\cite{dolev1986reaching,mendes2015multidimensional} to CUB spaces~\cite{roller2016poc,haettel2025link}.
Let us now fix $X$ a convexly uniquely bicombable (CUB) space, i.e, $X$ is a  metric space such that $X$, and each ball of $X$, admits a unique convex bicombing:

\begin{definition}[\cite{haettel2025link}]\label{def:bicombing}
A \textbf{convex bicombing} on $X$ is a map $\sigma: X \times X\times [0,1] \to X$, denoted $\sigma_t(x,y)$ for all $x,y \in X$ and $t \in [0,1]$, such that $\sigma$ is:
\begin{itemize}
\item \textbf{Geodesic}: for all $x, y \in X$, the map $t \in [0, 1] \mapsto \sigma_t(x, y)$ is a \textbf{constant speed} \textbf{ geodesic} from $\sigma_0(x, y) = x$ to $\sigma_1(x, y) = y$, i.e., $t \mapsto \sigma_t(x, y)$ is a shortest path from $x$ to $y$ such that $d(x, \sigma_t(x, y)) = t \times d(x,y)$ for all $t \in [0,1]$;
\item \textbf{Symmetric}: for all $x, y \in X$ and $t \in [0, 1]$, we have $\sigma_t(x, y) = \sigma_{1-t}(y, x)$;
\item \textbf{Affine}: for all $x, y \in X$ and $s, t \in [0, 1]$, we have $\sigma_{st}(x, y) = \sigma_s(x, \sigma_t(x, y))$;
\item \textbf{Convex}: for all $x, y, x', y' \in X$, the map $t \mapsto d(\sigma_t(x, y), \sigma_t(x', y'))$ is convex. 
\end{itemize}

A subspace $Y \subseteq X$ is \textbf{$\sigma$-convex} if for all $x,y \in Y$ and $t\in [0,1]$, $\sigma_t(x,y) \in Y$. The \textbf{$\sigma$-convex hull} of a subspace $Y \subset X$ is the smallest $\sigma$-convex set containing $Y$. If $Y$ is an Euclidean subspace of $X$, then $\sigma$-convexity on $Y$ agrees with Euclidean convexity, since the unique choice of bicombing on $Y$ is the Euclidean one. The \textbf{midpoint operator} $x \mid y = \sigma_{\frac12}(x,y)$ is defined for all $x,y \in X$. 
\end{definition}

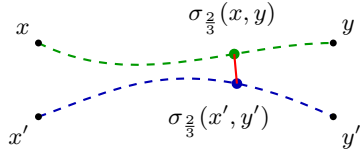
\begin{figure}[t]
    \centering
    \begin{tikzpicture}[scale=.65]
    % First path
    \begin{scope}[decoration={
        markings,
        mark=at position 0.667 with {
            \coordinate (P);
            \fill (0,0) circle (2pt) node[below , black] at (-.2,-.2) {$\sigma_{\frac 2 3}(x',y')$};
        }
    }]
    \draw[thick, blue!70!black, dashed, postaction={decorate}] 
        (0,.5) .. controls (1.7,1.1) and (3,1.9) .. (6,.5);
    \end{scope}

    \begin{scope}[decoration={
        markings,
        mark=at position 0.667 with {
            \coordinate (Q);
            \fill[green!60!black] (0,0) circle (2pt) node[above, black] at (0,.15) {$\sigma_{\frac 2 3}(x,y)$};
        }
    }]
    \draw[thick, green!60!black, dashed, postaction={decorate}] 
        (0,2) .. controls (2,1) and (4,2) .. (6,2);
    \end{scope}
    
    % Draw line between the two 2/3 points
    \draw[thick, red, -] (P) -- (Q);
    
    % % Add label on the center left of this line
    % \node[black, font=\small, fill=white, inner sep=2pt, rounded corners] 
    %     at ($(P)!0.44!(Q)$) [right=0.2cm] {$d(m,m')$};
    
    % Endpoints for first path
    \fill (0,.5) circle (2pt) node[below left] {$x'$};
    \fill (6,.5) circle (2pt) node[below right] {$y'$};
    
    % Endpoints for second path
    \fill (0,2) circle (2pt) node[above left] {$x$};
    \fill (6,2) circle (2pt) node[above right] {$y$};
    
\end{tikzpicture}
    \caption{Two geodesics in a convex bicombing, with corresponding points marked at $t = \frac{2}{3}$. The bicombing is convex: the red edge decreases in length from $t = 0$ to $t = 0.4$, then increases from $t = 0.4$ to $t = 1$.}
        \label{fig:convex}
\end{figure}

Intuitively, a space is CUB when two points can be joined by a canonical shortest path, and when this canonical choice satisfies a convexity condition, see Figure~\ref{fig:convex}, which forces the space to be non-positively curved. Several non-positively curved spaces are CUB, including Euclidean and CAT(0)\footnote{The sole property of CAT(0) spaces required in this article is that they are CUB.} spaces~\cite[Examples 1.3]{haettel2025link}. However, unlike CAT(0), there are combinatorial criteria for deciding whether a simplicial complex is CUB~\cite[Section 4]{haettel2025link}. In CAT(0), such criteria exist for cubical complexes only~\cite{gromov1987hyperbolic}.

\begin{remark}
\emph{Barycentric algebras}~\cite{stone1949postulates}, spaces equipped with a binary ``convex combination'' operator satisfying certain axioms, also offer a notion of convexity on metric spaces. CUB spaces cover different examples: not all CUB spaces are barycentric algebras. In particular, it is much easier to verify that the geometric realization of a simplicial complex is CUB rather than a barycentric algebra: we show in Section~\ref{section:CUB_from} that \emph{collapsible} simplicial complexes are CUB. 
\end{remark}

Let us now define the $\epsilon$-approximate agreement problem on $X$:

\begin{definition}[$\epsilon$-approximate agreement on CUB spaces]
Fix $\epsilon > 0$. A fixed number $n$ of processes $(P_i)_{i \in [n]}$ are each given a point $x_i \in X$. Each process must output a point $y_i \in X$, such that these conditions hold:
\begin{itemize}[noitemsep]
\item Validity: the outputs lie in the $\sigma$-convex hull of the inputs.
\item Agreement: the outputs are $\epsilon$-close to each other.
\end{itemize}
\end{definition}

Fix $\epsilon > 0$ and $n \in \mathbb N$. %Let $D$ denote the diameter of $X$.
We now propose a $N = \lceil \log_{1 - \frac{1}{2^{n-1}}}(\frac \epsilon D) \rceil$ steps protocol for solving $\epsilon$-approximate agreement on $X$ with $n$ processes in the layered immediate snapshot model, where $D$ denotes the maximal distance between a pair of inputs (or $1$ if all inputs are identical). The intuitive idea of computing barycenters at each round can still be used, even though the notions of barycenter and convex combinations are not well-defined in CUB spaces: only convex combinations of two points are defined.

\begin{protocol}[$\epsilon$-approximate agreement]\label{def:protocol_CUB}
At each round $k \in \mathbb N$, the processes will exchange their view and compute a new point: 
for all $i \in [n]$ and $k \in \mathbb N$, let $x^k_i$ denote the point of the $i$th process at the beginning of the $k$th step. Hence, $x^0_i = x_i$. 
At the start of the $k$th round, the processes proceed with the immediate snapshot protocol to exchange their points. Each process $i \in [n]$ then computes: 
$$x^{k+1}_i = f_n(p_1, \ldots, p_n) \coloneq p_1 \mid (p_2 \mid (\cdots (p_{n-1} \mid p_n ) \cdots))$$ where for all $j \in [n]$, $p_j$ denotes the point $x^k_j$ if the $j$th process view was received, of the current point $x^k_i$ otherwise. At the end of the last round $N$, each process $i \in [n]$ returns $y_i = x^N_i$.
\end{protocol}

\ifextended

\begin{remark}\label{remark:color_full} 
Protocol~\ref{def:protocol_CUB} is \emph{colored} in the sense that processes keep track of the ids of the processes from whom they receive information. These ids are used only in the the function $f$. Indeed, $f$ is asymmetric: not only because it assigns unbalanced weights to its inputs, but rather because CUB spaces are not necessarily barycentric algebras, hence the following associativity axiom might not hold:
\begin{equation}
    \sigma_{\lambda}(x, \sigma_{\mu}(y, z)) = \sigma_{\lambda\mu} \left( \sigma_{\frac{\lambda(1-\mu)}{1-\lambda\mu}}(x, y), z \right) \qquad (\lambda\mu \neq 1)\label{eq:bary}
\end{equation} 

To see an example, consider the tree $T$ rooted at $O$ with three edges $O$-$x$, $O$-$y$ and $O$-$z$. $T$ is a tree, i.e., a $1$-dimensional simplicial, or cubical complex. the geometric realization of a tree is always CAT(0), hence CUB, when endowed with the piecewise metric that assigning length $1$ to each edge. Now, for $\lambda = 0.5 = \mu$ and $\lambda\mu \neq 1$, we have: 
\begin{align*}
    &\sigma_{\frac 1 2}(x, \sigma_{\frac 1 2}(y, z)) = \sigma_{\frac 1 2}(x, O) \neq O & & \text{ but } & \sigma_{\frac 1 4} \left( \sigma_{\frac 1 3}(x, y), z \right) = O.
\end{align*}
Therefore, Equation~\eqref{eq:bary} does not hold, and $T$ is not a barycentric algebra. 
However, if a CUB space is also a barycentric algebras respectively to $\sigma$, then one can define a symmetric version of $f$, which makes the protocol colorless.   
\end{remark}
\else
\begin{remark}\label{remark:color}
Protocol~\ref{def:protocol_CUB} is \emph{colored} in the sense that processes keep track of the ids of the processes from whom they receive information. Indeed, $f$ is in general not invariant under permutation of its arguments. We show in~\cite[Remark 2]{albert2026solving} that if $X$ is also a barycentric algebra with respect to $\sigma$, then $f$ is invariant under permutation and the protocol becomes \emph{colorless}. We provide in addition an example of a CUB space that is not a barycentric algebra. 
\end{remark}
\fi

\ifextended
We now show that Protocol~\ref{def:protocol_CUB} is correct, i.e., that it satisfies validity and agreement, using the following lemmas.
\else
We show that Protocol~\ref{def:protocol_CUB} is correct, i.e., that it satisfies validity and agreement, using the following lemma, whose proof can be found in~\cite{albert2026solving}. 
\fi

\ifextended

\begin{lemma}\label{lemma:sum1}
For all $m$-tuples of points $(p_1, \ldots, p_m)$ and $(q_1, \ldots, q_m)$ in $X$:
\begin{equation}\label{eq:convex}
d(f_m(p_1, \ldots, p_m), f_m(q_1,\ldots q_m)) \leq {\sum_{i=1}^{m-1} \frac{1}{2^i}d(p_i,q_i) } + \frac{1}{2^{m-1}}d(p_m,q_m)
\end{equation}
\end{lemma}
\begin{proof}
Let us show the result by induction on $m > 0$: 
\begin{itemize}
    \item For $m = 1$, we have $f_1(x) = x$ for all points $x \in X$, so Equation~\eqref{eq:convex} holds.
    \item Suppose Equation~\eqref{eq:convex} holds for $m > 0$. Consider two tuples $(p_1, \ldots, p_{m+1})$ and $(q_1, \ldots, q_{m+1})$, and let $p'$ and $q'$ denote $p_m \mid p_{m+1}$ and $q_m \mid q_{m+1}$, respectively. We have the following inequality:
    \begin{align*}
         &d(f_{m+1}(p_1, \ldots, p_{m+1}), f_{m+1}(q_1,\ldots q_{m+1}))\\ &=  d(f_{m}(p_1, \ldots, p_{m-1},p'), f_{m}(q_1,\ldots q_{m-1},q')) \\ &\leq  {\sum_{i=1}^{m-1} \frac{1}{2^i}d(p_i,q_i) } + \frac{1}{2^{m-1}}d(p',q') \\%& \text{ by IH}\\
        &\leq  {\sum_{i=1}^{m-1} \frac{1}{2^i}d(p_i,q_i) } + \frac{1}{2^{m-1}}(\frac12d(p_m,q_m)+\frac12d(p_{m+1},q_{m+1}))\\% & \text{by convexity of $\sigma$} \\
        & = {\sum_{i=1}^{m} \frac{1}{2^i}d(p_i,q_i) } + \frac{1}{2^{m}}d(p_{m+1},q_{m+1})
    \end{align*}
    where the third step uses convexity of the bicombing $\sigma$ applied to $t= \frac12$. \qed
    \end{itemize}     
\end{proof}
\fi

\begin{lemma}\label{lemma:sum2}
For all $d >0$ and $m$-tuples $(p_1, \ldots, p_{m})$ and $(q_1, \ldots, q_{m})$ such that $d(p_i,q_i) \leq d$ for all $i\in [m]$ and $p_j = q_j$ for some $j \in [m]$, we have: \begin{equation}\label{eq:ineq}
        d(f_m(p_1,\ldots,p_m),f_m(q_1,\ldots, q_m)) \leq d \times (1 - \frac{1}{2^{m-1}}).
    \end{equation}
\end{lemma}
\ifextended
\begin{proof}
    If $j = m$, $d(f_m(p_1,\ldots,p_m),f_m(q_1,\ldots, q_m)) \leq {\sum_{i=1}^{m-1} \frac{1}{2^i}d } = d \times (1 - \frac{1}{2^{m-1}})$. Otherwise:
\begin{align*}
    d(f_m(p_1,\ldots,p_m),f_m(q_1,\ldots, q_m)) &\leq {\sum_{i=1}^{j-1} \frac{1}{2^i}d } + {\sum_{i=j+1}^{m-1} \frac{1}{2^i}d } + \frac{1}{2^{m-1}}d\\
    &\leq {\sum_{i=1}^{j-1} \frac{1}{2^i}d } + {\sum_{i=j+1}^{m-1} \frac{1}{2^{i-1}}d } + \frac{1}{2^{m-1}}d \\
    & = {\sum_{i=1}^{m-1} \frac{1}{2^i}d } = d \times (1 - \frac{1}{2^{m-1}}).
    \end{align*} \qed
\end{proof}
\fi

\begin{theorem}
For any CUB space $X$, and any $\epsilon>0$, Protocol~\ref{def:protocol_CUB} solves $\epsilon$-approximate agreement on $X$.
\end{theorem}
\begin{proof}
Let us show that protocol~\ref{def:protocol_CUB} satisfies validity and agreement.
For validity, any $\sigma$-convex subspace of $X$ is stable under the midpoint operator.
For agreement, let us show that after the $k$th round, the maximum distance $d_{k+1}$ between two points $x^{k+1}_i$ and $x^{k+1}j$ is less or equal to $d_{k} \times (1 - \frac{1}{2^{n-1}})$. It suffices to apply Inequality~\eqref{eq:ineq} given by Lemma~\ref{lemma:sum2} at each round. To do so, we must prove that at each step $k$, any two processes $i \neq j \in [n]$ both received the view of a common process $k \in [n]$. However, if process $i$ wrote before $j$, then both processes received the view of process $i$, and symmetrically. Initially, the processes' points are $D$-close to each other. Therefore, after $k$ rounds, the processes points are $D(1 - \frac{1}{2^{n-1}})^k$ close to each others. Thus, after the last step at round $N = \lceil\log_{1 - \frac{1}{2^{n-1}}}(\frac \epsilon D)\rceil$, the points are $\epsilon$-close to each other. \qed
\end{proof}

\subsection{CUB spaces from finite collapsible complexes}\label{section:CUB_from}

In this section, we show that $\mathcal C$ can be endowed with a structure of CUB space that is \emph{bounded}, in the sense that $\mathcal C$ has finite diameter.
Since $\mathcal C$ is finite, its geometric realization $|\mathcal C|$ is well-defined and finite-dimensional. We denote by $d$ its dimension. We can apply the following result:

\begin{theorem}[{\cite[Theorem 4]{AdiprasitoF21}}]\label{th:adiprasito}
Any finite-dimensional collapsible simplicial complex is PL-homeomorphic to a cubical complex which is
CAT(0) when endowed with the piecewise flat equilateral pseudometric.
\end{theorem}

Explicitly, there exists a cubical complex $\mathcal D$ that is CAT(0) when endowed with the piecewise flat equilateral pseudometric, and such that $|\mathcal C|$ is piecewise linear-homeomorphic to $\mathcal D$. That is, there are simplicial subdivisions $\mathcal K$ and $\mathcal L$ of $|\mathcal C|$ and $\mathcal D$, respectively, along with an homeomorphism $\mathcal K \simeq \mathcal L$ such that each simplex of $\mathcal K$ is mapped linearly to a simplex of $\mathcal L$. 

\begin{figure}[t]
    \centering
    \begin{tikzpicture}[scale=.7]

\coordinate (A) at (0,0);       % front-bottom-left
\coordinate (B) at (8,0);       % front-bottom-left

\draw[thick, dotted] (0,0) -- (.5,0);
\draw[thick] (.5,0) -- (B);

\filldraw[black] (A) circle (2pt);
\filldraw[black] (B) circle (2pt);
\fill[white,draw=black, thick] (4,0) circle (2pt);
\fill[white,draw=black, thick] (2,0) circle (2pt);
\fill[white,draw=black, thick] (1,0) circle (2pt);
\fill[white,draw=black, thick] (.5,0) circle (2pt);

\node at (6,  .4) (a) {$1$};
\node at (3,  .4) (b) {$1$};
\node at (1.5,.4) (c) {$1$};
\node at (.75,.4) (d) {$1$};

\end{tikzpicture}
    \caption{An infinite subdivision of the unit edge in $\mathbb R$. The pullback metric from the proper cubical complex assigns length $1$ to each subdivided edge, making the subdivided edge unbounded. }
    \label{fig:infinite}
\end{figure}
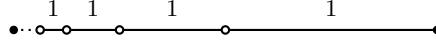

It follows that the pullback pseudometric\footnote{For a map \(f\colon X\to Y\) and a pseudometric \(d_Y\) on \(Y\), the pullback pseudometric on \(X\) is
$
d_X(x,x') = d_Y(f(x), f(x')).
$} endows $|\mathcal C|$ with a CAT(0) pseudometric, which agrees with the usual topology on $|\mathcal C|$. However, for $|\mathcal C|$ to be CUB, this pseudometric must be a metric. This holds if and only if the subdivision $\mathcal{K}$ of $|\mathcal C|$ is finite, see for example Figure~\ref{fig:infinite}. 
\ifextended
Fortunately, we can check that the proof of Theorem~\ref{th:adiprasito} produces a finite subdivision $\mathcal K$.

\begin{proposition}
    The subdivision $\mathcal K$ produced by Theorem~\ref{th:adiprasito} is a finite simplicial complex.
\end{proposition}
\begin{proof}
    
In this proof, all the complexes are understood as \emph{geometric}.
The proof of Theorem~\ref{th:adiprasito} writes $\mathcal{C}$ as a finite increasing union of simplicial subcomplexes:
\[
\mathcal C_0 \subset \mathcal C_1 \subset \cdots \subset C_N = \mathcal{C},
\qquad \mathcal C_0 = \{*\},
\]
where each $\mathcal C_{n+1}$ is obtained from $\mathcal C_n$ by attaching a $k$-simplex $\delta$ along $\gamma = \delta \cap C_n$.
The construction builds, by induction, a CAT(0) cubical subdivision $\mathcal D_n$ of $\mathcal C_n$, such that $\mathcal D_n$ is PL-homeomorphic to $\mathcal C_n$. Let us show by induction that this PL-homeomorphism is done through finite simplicial subdivisions $\mathcal K_n$ and $\mathcal L_n$ of $\mathcal C_n$ and $\mathcal D_n$, respectively. The result holds for $C_0 = \{*\}$. For the inductive step, the simplicial subcomplex $\gamma$ is sent by the PL-homeomorphism to a cubical $k$-disk $\Gamma$ in $\mathcal K_n$, and by~\cite[Proposition 10]{AdiprasitoF21}, there is a cubical $(k+1)$-disk $\Delta$ with $\Gamma$ in its boundary, that is PL-homeomorphic to $\delta$, such that the subdivisions associated to the PL-homeomorphism agree with $\mathcal K_n$ and $\mathcal L_n$. $\mathcal K_{n+1}$ and $\mathcal L_{n+1}$ are therefore obtained by gluing these compatible subdivisions of $\delta$ and $\Delta$. Thus, we have to show that $\Delta$ and the subdivisions associated to the PL-homeomorphism $\Delta \simeq \delta$ are finite. However, this is immediate as~\cite[Proposition 10]{AdiprasitoF21} build $\Delta$ by applying a series of modification to $\Gamma \times [0,1]$, which preserves our finiteness hypothesis. 
\end{proof}

\else
Fortunately, we check in~\cite{albert2026solving} 
that the proof of Theorem~\ref{th:adiprasito} produces a finite (as a complex) subdivision $\mathcal K$.
\fi
Therefore, $|\mathcal C|$ is CUB, and we can consider approximate agreement on the CUB space $|\mathcal C| \simeq \mathcal K$. This CUB structure has a few consequences on the simplices of $\mathcal K$:
\begin{itemize}
    \item each simplex $\kappa$ of $\mathcal K$ is $\sigma$-convex when regarded as a subspace of $|\mathcal C|$;
    \item the induced metric on a simplex $\kappa$ of $\mathcal K$ is the pullback through $\mathcal K \to \mathcal L$ of the usual Euclidean metric on a geometric simplex of $\mathcal L$;
    \item the shortest path in $\mathcal K$ between two points in a simplex $\kappa$ of $\mathcal K$ lies in $\kappa$.
\end{itemize}
On simplices of $|\mathcal C|$, however, none of these properties necessarily hold, see for example Figure~\ref{fig:CUB}.

\begin{figure}[t]
    \centering
        \begin{tikzpicture}[scale = .65]
\fill[blue!30, draw=black, thick] (-1,0) -- (3,0) -- (1,2) -- cycle;
\draw[dashed, thick] (1,2) -- (1,0);

\filldraw[black] (-1,0) circle (2pt);
\filldraw[black] (3,0) circle (2pt);
\filldraw[black] (1,2) circle (2pt);
\fill[blue!30, draw=black, thick] (1,0) circle (2pt);

\draw (0,0) node[cross,thick, red] {};
\draw (2,0) node[cross,thick, rotate=0, red] {};

\draw[dotted, thick, red] (0,0) -- (1,.37) -- (2,0);

\end{tikzpicture}
     ~ \Large $\simeq$ ~
        \begin{tikzpicture}[scale=.65]
\def\x{1.9}

\fill[blue!30, draw=black, thick] (0,0) -- (\x,0) -- (\x,\x) -- (0,\x) -- cycle;
\draw[dashed, thick] (0,0) -- (\x,\x);

\filldraw[black] (0,0) circle (2pt);
\filldraw[black] (\x,0) circle (2pt);
\filldraw[black] (0,\x) circle (2pt);
\filldraw[black] (\x,\x) circle (2pt);

\draw (\x*0.5,0) node[cross,thick, rotate=0, red] {};
\draw (0, \x*0.5) node[cross,thick, rotate=0, red] {};

\draw[dotted, thick, red]  (\x*0.5,0) -- (0, \x*0.5);

\end{tikzpicture}
    ~ ~ 
    \caption{A geometric simplicial complex $\mathcal C$ that is PL-homeomorphic to a CAT(0) cubical complex $\mathcal D$. The subdivision induced by the PL-homeomorphism is shown with dashed edges and hollow vertices. Two points are depicted, along with the shortest path between them in the piecewise flat equilateral metric of $\mathcal D$ and in the pullback metric of $\mathcal C$. In particular, the bottom edge of $\mathcal C$ is not  $\sigma$-convex with this metric.}
    \label{fig:CUB}
\end{figure}
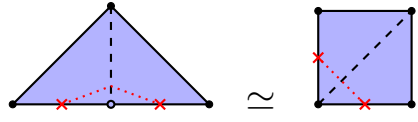

\subsection{Recovering vertices from approximate agreement on \texorpdfstring{$|\mathcal C|$}{C}}\label{section:recovering}

Let us now reduce approximate agreement on the CUB space $\mathcal K$ to simplex agreement on $\mathcal K$. Suppose given a fixed number $n$ of processes with input vertices $(x_i)_{i \in [n]}$ of $\mathcal K$. For all $\epsilon > 0$, we can apply the protocol for approximate agreement on $\mathcal K$ to obtain output points $(y_i)_{i \in [n]} $ that are $\epsilon$-close to each other. However, there is no guarantee that these points are vertices of $\mathcal K$. 
In the following, we show that for small enough values of $\epsilon$, each process can compute an output vertex of $\mathcal K$, without additional communication, and such that validity and agreement are satisfied.
We first expose some general geometry results.

%\begin{definition}
    The distance $d(A,B)$ between two subsets $A$ and $B$ of a metric space is the infimum $\inf\{d(a,b) \mid a \in A, b \in B\}$. If $A$ and $B$ are nonempty and compact, this infimum is attained. In $\mathcal K$, any simplices $\kappa, \tau$ are compact; hence $\kappa \cap \tau \neq \varnothing$ if and only if $d(\kappa, \tau) = 0$.
Similarly, we define the distance $d(x,A)$ between a point $x$ and a subset $A$ of a metric space as $\inf\{d(x,a) \mid a \in A\}$. When this infimum is attained by a point (i.e. for a compact $A$), we call this point a \textbf{projection} of $x$ on $A$. When this point is unique, for example when projecting a point of a geometric simplex into a face, we denote it by $p_A(x)$. 
%\end{definition}

Since $\mathcal K$ is finite, the minimum distance $d_{min}$ between two non-intersecting simplices of $\mathcal K$ is well defined. \ifextended\else We prove the following geometry lemmas in~\cite{albert2026solving}. \fi Intuitively, they show that if a point $x$ in $\mathcal K$ is close enough to two simplices $\kappa$ and $\tau$, then they intersect and $x$ is close to their intersection, in a uniform way.

\begin{lemma}\label{lemma:corner1}
    Let $\Delta$ be a simplex in $\mathcal K$, and let $\kappa$ and $\tau$ be sub-simplices of $\Delta$. There exists $C_\Delta(\kappa, \tau) \geq 1$ such that for all $\delta > 0$ and points $x$ in $\Delta$, if $\delta < d_{min}/2$ and $x$ is $\delta$-close to $\kappa$ and $\tau$, then $\kappa \cap \tau \neq \emptyset$ and $x$ is $C_\Delta(\kappa, \tau)\delta$-close to $\kappa \cap \tau$.
\end{lemma}
\ifextended
\begin{proof}
$d(\kappa, \tau) \leq d(\kappa, x) + d(x, \tau) < d_{min}$ so $\kappa$ intersects $\tau$.
    If $\kappa$ is contained in $\tau$, or the reverse, $C_\Delta(\kappa, \tau) = 1$ suffices. Otherwise, we have $d(x, \kappa \cap \tau) =  d(x,p_{\kappa\cap \tau}x)$, and the inequality:
    \begin{align*}
         d(x,p_{\kappa\cap \tau}x) &\leq d(x, p_{\kappa}x) + d(p_\kappa x, p_\tau p_\kappa x) + d(p_\tau p_\kappa x, p_{\kappa\cap \tau}x)\\
          &= d(x, \kappa) + d(p_\kappa x, \tau) + d(p_\tau p_\kappa x, p_{\kappa\cap \tau}x)\\
          &\leq d(x, \kappa) + d(p_\kappa x, x) + d(x,\tau) + d(p_\tau p_\kappa x, p_{\kappa\cap \tau}x)\\
          &= 2d(x, \kappa) + d(x, \tau) + d(p_\tau p_\kappa x, p_{\kappa\cap \tau}x)
        %&\leq 3\delta + d(p_\tau p_\kappa x, p_{\kappa\cap \tau}x)
        % &\leq d(x,\kappa) + d(p_\kappa x,\tau) + || p_\tau p_\kappa x -p_{\kappa\cap \tau}x||\\
        % &\leq d(x,\kappa) + d(p_\kappa x,\tau) + || p_\tau p_\kappa x -p_{\kappa\cap \tau}x||\\
    \end{align*}
        By classical results of best approximation theory, see for example~\cite[Theorem 9.33]{deutsch2001best}, there exists a constant $ 0 \leq c(\tau,\kappa) \leq 1$ depending only on the angle between $\kappa$ and $\tau$ such that: 
        $$d(p_\tau p_\kappa x, p_{\kappa\cap \tau}x) \leq c(\kappa, \tau) d(x, p_{\kappa \cap \tau} x).$$
        In addition, $c(\tau,\kappa) = 1$ if and only $\tau$ is contained in $\kappa$, or the reverse.
        In our case, $c(\tau,\kappa) < 1$ and we have:
        \begin{align*}
            d(x,p_{\kappa\cap \tau}x) &\leq 2d(x, \kappa) + d(x, \tau) + d(p_\tau p_\kappa x, p_{\kappa\cap \tau}x)\\
            &\leq 3\delta + c(\kappa, \tau) d(x, p_{\kappa \cap \tau} x),
        \end{align*}
        hence $d(x, p_{\kappa \cap \tau} x) \leq \frac{3 \delta}{1 -c(\kappa, \tau)}$ and we thus set $C_\Delta(\kappa, \tau) = \frac{3}{1 -c(\kappa, \tau)}$. \qed
\end{proof}
\fi

Since $\mathcal K$ is finite, we can consider $C_m = \max_{C \in \mathcal{K},\, \kappa,\tau \subseteq C} C_\Delta(\kappa, \tau)$.

\begin{lemma}\label{lemma:corner2}
    There exists $C \geq 1$ such that for all simplices $\kappa$ and $\tau$ in $\mathcal K$, $d_{min} / 2C_m > \delta > 0$ and point $x$ in $\mathcal K$, if $x$ is $ \delta$-close to $\kappa$ and $\tau$, then $\kappa \cap \tau \neq \emptyset$ and $x$ is $C\delta$-close to $\kappa \cap \tau$.
\end{lemma}
\ifextended
\begin{proof}
If $x$ is a vertex of $\mathcal K$, then $d(\{x\},\kappa) = d(x,\kappa) < d_{min}$ so $\{x\}$ intersects $\kappa$, i.e. $x \in \kappa$. Symmetrically, $x \in \tau$ and the result holds trivially.
Otherwise, $x$ lies in the interior of some unique simplex $\Delta$, hence $d(x, \Delta)=0$. Now, $d(\Delta, \kappa) \leq d(x,\kappa) < d_{min}$ so $\Delta \cap \kappa \neq \emptyset$ and $d(x, \Delta \cap \kappa) \leq C_m\delta < d_{min} / 2$. Symmetrically, $d(x, \Delta \cap \tau) \leq C_m\delta  < d_{min} /2$. Therefore, $\Delta \cap \kappa$ and $\Delta \cap \tau$ intersects. In particular, $\kappa \cap \tau \neq \emptyset$. Moreover, we have: $$d(x, \kappa \cap \tau) \leq d(x, (\Delta \cap \kappa) \cap (\Delta \cap \tau)) \leq C_m^2\delta.$$ Therefore, having $C = C_m^2$ suffices. \qed
\end{proof}
\fi

Let us now show that for all values $ \frac{d_{min}}{2C_m(3C)^d}> \epsilon > 0$ and $n$ processes with input points $(y_i)_{i \in [n]}$ in $\mathcal K$ that are $\epsilon$-close to each other, each process $i$ can compute a vertex $z_i$ in $\mathcal K$ without communication, such that the output vertices $(z_i)_{i \in [n]}$ span a simplex of $\mathcal K$. Consider the following $d$-round protocol: 

\begin{protocol}[Face projection]\label{def:protocol}
Each process with input $x$ computes:\\
%\begin{minted}[escapeinside=||,mathescape]{py}
for $i$ in $[d, d-1, \ldots, 1]$:\\
\hspace*{2em} if $x$ is in the interior of a simplex $\kappa$ of dimension $i$:\\
        \hspace*{4em}$\tau \leftarrow $ a face of $\kappa$ closest to $x$\\
        \hspace*{4em}$x \leftarrow p_{\tau}x $\\
return $x$
%\end{minted}
\end{protocol}

Protocol~\ref{def:protocol} returns vertices of $\mathcal K$: for each process with input $x$, at the beginning of each round, $x$ lies in the interior of a simplex of dimension at most $i$, hence is necessarily in the interior of an edge at the beginning of some iteration, and then projects into a vertex of $\mathcal K$. \ifextended We now prove its correctness: \else We prove its correctness in~\cite{albert2026solving}: \fi

\begin{proposition}\label{prop:common_simplex}
If the input points $(y_i)_{i \in [n]}$ in $\mathcal K$ are $ \frac{d_{min}}{2C_m(3C)^d} > \epsilon$-close to each other, then the output vertices $(z_i)_{i \in [n]}$ of Protocol~\ref{def:protocol} lie in the same simplex.% of $\mathcal K$. 
\end{proposition}
\ifextended
\begin{proof}
Let us show by induction that after the $k$th round, the processes' points lie either in the same simplex (boundary included), or are $\epsilon(3C)^k$-close to each other. For the first round $k = 0$, the points are indeed $\epsilon$-close. Suppose the result holds for $d-1 > k \geq 0$. If all the points already lie in the same simplex $\kappa$, then whether or not they project, they remain in $\kappa$. Otherwise, if no points are projected this round, then they remain $\epsilon(3C)^k \leq \epsilon(3C)^{(k+1)}$-close to each other. Otherwise, for all points $a$ lying in the interior of a simplex $\kappa$ of dimension $d-k$ ($a$ will be projected on a closest face of $\kappa$):

Since all points are not in $\kappa$, there is a point $b$ lying in a simplex $\tau \nsubseteq \kappa$ of dimension less or equal than $ d-k$. By hypothesis, $a$ is $$\epsilon(3C)^k < \frac{d_{min}}{2C_m(3C)^{(d-k)}} \leq \frac{d_{min}}{2C_m}$$ close to $b$, hence to $\kappa$ ($a \in \kappa$) and $\tau$ ($b \in \tau$). Therefore, Lemma~\ref{lemma:corner2} applies and  $a$ is $\epsilon(3C)^kC$-close to $\kappa \cap \tau \neq \emptyset$. Since $\kappa \cap \tau$ is a sub-simplex of $\kappa$ of dimension at most $d- k -1$, $a$ is $\epsilon(3C)^kC$-close to at least one face of $\kappa$, and a projection of $a$ will be $\epsilon(3C)^kC$-close to $a$.

The above reasoning holds for any two points $a$ and $b$. Let $a'$ and $b'$ denote their projections if they project, or the original points otherwise. We finally have:
\begin{align*}
    (a',b') &\leq d(a',a) + d(a,b) + d(b, b')\\
            &\leq \epsilon(3C)^kC + \epsilon(3C)^k + \epsilon(3C)^kC\\
            &\leq \epsilon(3C)^{k+1}.
\end{align*}
Finally, let us show that after the last round, the output vertices $(z_i)_{i \in [n]}$ lie in a common simplex of $\mathcal K$. By the property we just proved, either this is immediate, or the vertices are $d_{min} > \epsilon(3C)^{(d-1)}$-close to each other. By definition of $d_{min}$ applied to the simplices $(\{z_i\})_{i \in [n]}$, all the $z_i$ are equal to a common vertex $z$, hence lie in the simplex $\{z\}$. \qed
\end{proof}
\fi

We can now describe our reduction from approximate agreement on CUB spaces to simplex agreement on $\mathcal K$, and prove Theorem~\ref{th:snd}. 
\begin{protocol}[Simplex agreement on $\mathcal K$]\label{def:protocol_sub}
    Let $\epsilon$ be the largest power of $1/2$ lesser than $\frac{d_{min}}{2C_m(3C)^d}$. Given our $n$-processes with input points $(x_i)_{i\in [n]}$ in $\mathcal K$, first compute points $(y_i)_{i\in [n]}$ in $\mathcal K$ that are $\epsilon$-close to each other, using our protocol for $\epsilon$-approximate agreement in $\mathcal K$. Then, apply Protocol~\ref{def:protocol} with the input points $(y_i)_{i\in [n]}$ to get our final output vertices $(z_i)_{i \in [n]}$.
\end{protocol}
\begin{proof}[of Theorem~\ref{th:snd}]
Let us show validity and agreement. For validity, if the $(x_i)_{i \in [n]}$ span a simplex $\kappa$ of $\mathcal K$, then since $\kappa$ is $\sigma$-convex, the  points $(y_i)_{i \in [n]}$ lie in $\kappa$, and so do the output vertices $(z_i)_{i \in [n]}$, since a point in $\kappa$ cannot exit $\kappa$ by projecting onto a sub-simplex. For agreement, this is Proposition~\ref{prop:common_simplex}.  \qed
\end{proof}

We then obtain Corollary~\ref{cor:weaker} from a reduction of Protocol~\ref{def:protocol_sub}: 

\begin{protocol}[Simplex agreement on $\mathcal C$ under weak validity]

Given $n$ processes with inputs vertices $(x_i)_{i\in [n]}$ of $\mathcal C$ (hence of $\mathcal K)$, start by solving simplex agreement on $\mathcal K$: each process $i \in [n]$ computes an output vertex $y_i$ of the subdivision $\mathcal K$. If $y_i$ is also a vertex of $\mathcal C$, the $i$th process outputs $y_i$. Otherwise, it outputs any vertex of the simplex of $|\mathcal C|$ containing $y_i$ in its interior. 
\end{protocol}
\begin{proof}[of Corollary~\ref{cor:weaker}]
Let us prove agreement and weak validity of simplex agreement on $\mathcal C$. For weak validity, if the $(x_i)_{i \in [n]}$ are the same vertex $v$ of $\mathcal{C}$, then the points $(y_i)_{i \in [n]}$ outputted by the protocol for approximate agreement on CUB spaces are also all $v$, and so are the output vertices $(z_i)_{i \in [n]}$, since no projection occurs.
For agreement, let $(z_i)_{i \in [n]}$ be vertices of the same simplex of $\mathcal K$. It suffices to observe that all the points $z_i$ that are not vertices of $\mathcal C$ are in the interior of the same simplex of $\mathcal C$.  \qed
\end{proof}

Since the simplices of $\mathcal C$ are not necessarily $\sigma$-convex, in general even if the input vertices $(x_i)_{i \in [n]}$ in $\mathcal C$ span a simplex $\kappa$ of $\mathcal C$, the points $(y_i)_{i \in [n]}$ outputted by the protocol for approximate agreement on the CUB spaces $|\mathcal C|$ are not necessarily contained in $\kappa$ anymore, and validity cannot be ensured.

\section{Conclusion}
\label{sec:conclusion}
Approximate agreement is one of the oldest problems in distributed computing, yet the spaces on which it is solvable still lack an exact characterization.
%We have studied the solvability of several approximate agreement tasks. 
In the first part we considered discrete spaces. We gave a complete topological characterization of the solvability of simplex agreement, which is undecidable already for three processes, but decidable for collapsible and shellable complexes. We then discussed implications for clique agreement, some new and some previously known.
In the second part, we introduced the general task of approximate agreement on CUB spaces, and an explicit algorithm for solving this problem, which we applied to simplex agreement, providing explicit protocols for solving two weaker versions of approximate agreement on collapsible simplicial complexes.

There are many interesting questions left.
%Our results assume a wait-free model, but we could extend them to models where at most $t$ processes may crash, for instance via BG simulation~\cite{herlihy2013distributed}.
The usefulness of shellability for clique agreement is still unclear: no family of graphs is known whose clique complex is shellable, yet is not covered by the previous combinatorial criteria.
%For instance, Tur\'an graphs do have shellable clique complexes, but we conjecture that the solvability of clique agreement on these clique complexes is completely determined by the clique containment condition together with the existence of AER labelings.
Similarly, the $k$-clique containment condition gives an impossibility bound depending on some chromatic number~\cite{Liu22}, but its relationship to $k$-connectivity is not obvious.

Approximate agreement on CUB spaces is always solvable, generalizing classical and multidimensional approximate agreement but not generalized approximate agreement~\cite{herlihy1993asynchronous}, which allows for holes in the space. Therefore, extending our continuous framework to formally prove the claim in~\cite{herlihy1993asynchronous} would be valuable.
%of interest. However, a metric space with topological holes cannot be CUB.
%We could mimic the setting of~\cite{herlihy1993asynchronous} by restricting to a subset $A$ of allowed points in a CUB space $X$, and using their definition of holes. This is unsatisfactory, however, as it captures only $0$-dimensional holes. 
Finally, 
the weak validity condition in Corollary~\ref{cor:weaker} is not well understood. Its converse fails: $2$-set agreement for $3$ processes is a special case of simplex agreement on the hollow triangle, which is not collapsible, yet the problem is solvable under weak validity.
More generally, studying the solvability of set agreement in different spaces would be interesting.

\begin{credits}
\subsubsection{\ackname}
The author thanks Zhenghao Hu, J\'er\'emy Dubut, J\'er\'emy Ledent and Eric Goubault for the numerous interactions during the preparation of this work, and the anonymous reviewers for their valuable comments and suggestions.

\subsubsection{\discintname}
The authors declare that they have no known competing financial interests or
personal relationships that could have appeared to influence the work reported
in this paper.

\end{credits}

\newpage

\bibliographystyle{splncs04}
\bibliography{references}

\end{document}